\documentclass[11 pt]{article}

\usepackage{amsfonts,amssymb}
\usepackage{color}
\usepackage{amsthm, amssymb, latexsym, amsmath,  
}
\usepackage[margin=1.2 in]{geometry}

\newtheorem{theorem}{Theorem}[section]

\newtheorem{lemma}[theorem]{Lemma}

\newtheorem{proposition}[theorem]{Proposition}

\newtheorem{definition}{Definition}

\numberwithin{theorem}{section} \numberwithin{equation}{section}

\newcommand{\nc}{\newcommand}

\nc{\be}{\begin{equation}} \nc{\la}{\label} \nc{\ba}{\begin{array}}
\nc{\ea}{\end{array}} \nc{\bs}{\begin{split}} \nc{\es}{\end{split}}

\newcommand{\R}{\mathbb{R}}

\newcommand{\cA}{{\cal{A}}}

\newcommand{\cH}{{\cal{H}}}

\newcommand{\imsgap}{{\gamma}}

\newcommand{\vphi}{{\varphi}}

\newcommand{\bet}{{\beta}}      
\newcommand{\al}{{\alpha}}
\newcommand{\del}{{\delta}}
\newcommand{\lam}{{\lambda}}           

\newcommand{\Om}{\Omega}                

\newcommand{\s}{{\sigma}}

\renewcommand{\H}{H}
\newcommand{\E}{E}
\newcommand{\Einfty}{E(\infty)}

\nc{\G}{\Gamma} \nc{\g}{\gamma} \nc{\Omt}{\tilde{\Omega}}
\nc{\ta}{\tau} \nc{\w}{\omega} \nc{\io}{\iota} \nc{\h}{\theta}
\nc{\Si}{\Sigma}

\nc{\bP}{\bar{P}} \nc{\bQ}{\bar{Q}} \nc{\bL}{\bar{L}}

\newcommand{\DETAILS}[1]{}

\nc{\ran}{\rangle}
\nc{\lan}{\langle}

\newcommand{\Ran}{\operatorname{Ran}}
\newcommand{\Tr}{\operatorname{Tr}}

\newcommand{\supp}{\operatorname{supp}}
\newcommand{\rank}{\operatorname{rank}}

\newcommand{\one}{\mathbf{1}}

\nc{\bfone}{{\bf 1}}
\nc{\1}{{\bf 1}}

\newcommand{\ls}{\lesssim}
\newcommand{\dotge}{\dot \ge}
\newcommand{\dotle}{\dot \le}

\newcommand{\n}{\nabla}
\newcommand{\p}{\partial}

\newcommand{\at}{\operatorname{at}}

\begin{document}
\title{Long Range Behaviour of van der Waals Force}

\author{Ioannis Anapolitanos \thanks{Dept.~of Math.,
Univ. of Stuttgart, Stuttgart, Germany; Supported by DFG under Grant
GR 3213/1-1.}\ \qquad \qquad Israel Michael Sigal
\thanks{ Dept.~of Math., Univ. of Toronto, Toronto, Canada; Supported by NSERC Grant No. NA7901}
}
\bigskip

\bigskip

\bigskip

\date{To Elliott Lieb with admiration and  friendship\\ \bigskip \bigskip
August 24, 2013}

\maketitle

\begin{abstract}
   We prove van der Waals - London's law of decay of the van der Waals force for a collection of neutral atoms at large separation.
\end{abstract}


\section{Introduction}
 Van der Waals force between atoms and molecules plays a fundamental role in
quantum chemistry, physics and material sciences. For instance, it
defines the chemical character of many organic compounds and enable
geckos - which can hang on a glass surface using only one toe - to
climb on sheer surfaces (see the entry "Van der Waals force" in Wikipedia). 
 It explains why water  condenses from vapor as well as  many properties of molecular compounds, including crystal structures (e. g. the shapes of snowflakes), melting points, boiling points, heats of fusion and vaporization, surface tension, and densities. It forces gigantic molecules like enzymes, proteins, and DNA into the shapes required for biological activity (see \cite{Sen}).

 A microscopic
explanation of this force was given by F. London soon after the
discovery of quantum mechanics and was one of its early triumphs.
  This heuristic explanation showed that this force has a
universal behavior at large distances - it decays as the inverse
sixth power of the distance between atoms. This behavior was
confirmed in \cite{LT} by proving - through a sophisticated test
function construction - an upper bound. Our goal in this paper is to
provide a complete proof of the  van der Waals - London decay law
for atoms.

Let $-e$ and $m$ denote the electron charge and mass. Consider a
system of $M$ multielectron atoms which we call a molecule though we
do not assume binding between atoms. In the units where  $\hbar=1$
the hamiltonian of the system is
\begin{equation*}
H_{\rm mol}= \sum_{i=1}^{N} (- \frac{1}{2
m}\Delta_{x_i}-\sum_{j=1}^{M} \frac{e^2
Z_j}{|x_i-y_j|})+\sum_{i<j}^{1,N}
\frac{e^2}{|x_{i}-x_{j}|}-\sum_{j=1}^M \frac{1}{2 m_j} \Delta_{y_j}
+\sum_{i<j}^{1,M} \frac{Z_{i} Z_{j}e^2}{|y_{i}-y_{j}|}.
\end{equation*}
  Here $N$ is the total number of electrons, $x_i,y_i\in \R^3$ denote
the coordinates of the electrons and the nuclei, respectively, $e
Z_j$ is the charge of the $j$-th nucleus, $m_j$ is the mass of
the $j$-th nucleus and $\Delta_{x_j}$ is the Laplacian and gradient acting on the coordinate $x_j$. We consider a system of neutral atoms so we must
have $\sum_{j=1}^M Z_j=N$.
 The operator $H_{\rm mol}$ acts on the subspace
$\cH_{\textrm{fermi}}^{\textrm{mol}}$ of the space
$L^2(\R^{3(N+M)})$, which accounts for the fact that the electrons
are identical particles and are fermions and therefore they obey the
Fermi-Dirac statistics. Also, possibly, some of the nuclei are
identical and obey either the Fermi-Dirac or Bose-Einstein
statistics (more details are given below).
 For $M=1$, $H_{\rm mol}$ is the Hamiltonian of the atom with the nucleus of charge $eZ=eZ_1$.

To define the interaction energy, let $a$ be the decomposition of the molecule into neutral atoms and/or ions, each with its own nucleus,  and let $H_a$
be the sum of the corresponding atomic or ionic Hamiltonians (see the next section for precise definitions).  We define the energy  $\Einfty$ of the system with the atoms or ions infinitely far from each other, as
\begin{equation}\label{Einftydef}
\Einfty=\min_a E_a,\ \quad \mbox{where}\ \quad  E_a:=\inf \s(H_a).
\end{equation}
(Of course, one expects that $\Einfty$ is obtained by taking the minimum over the atomic decompositions only, see the discussion below.) 
\DETAILS{Let $E_{m, n}$ denote the  ground state energy
of the ion with a nucleus of charge $e Z_m$ and $Z_m-n$ electrons, so that $E_{m, 0}=E_{m}$ is the ground state energy 
of the atom corresponding to the $m$-th nucleus of charge $Z_m$, and}
Furthermore, let $y=(y_1,...,y_M)$ be the collection of the nuclear
co-ordinates. 
The interaction (or cohesive) energy $W(y)$ between atoms in this
system is defined as
\begin{equation}\label{interene}
W(y):=E(y)-\Einfty,
\end{equation}
where  $E(y)$ is the ground state energy of the system with
positions of the nuclei fixed, i.e.  the ground state energy of the
hamiltonian
\begin{equation}\label{Hy}
H_N(y)=\sum_{i=1}^{N} (- \frac{1}{2m} \Delta_{x_i}-\sum_{j=1}^{M}
\frac{e^2 Z_j}{|x_i-y_j|})+\sum_{i<j}^{1,N}
\frac{e^2}{|x_{i}-x_{j}|}+\sum_{i<j}^{1,M} \frac{Z_{i} Z_{j}
e^2}{|y_{i}-y_{j}|},
\end{equation}
 acting  on the subspace $\cH_{\textrm{fermi}}$ of the space $L^2(\R^{3N})$, which accounts for  the Fermi-Dirac statistics
of electrons.
\DETAILS{, and   $\Einfty$ is the energy of the system with the
atoms infinitely far from each other,
\begin{equation}\label{Einftydef}
\Einfty=\min\{ \sum_{m=1}^M E_{m, n_m}| \sum_{m=1}^M n_m=0\}.
\end{equation}
(Of course, one expects that $\Einfty=\sum_{m=1}^{ M} E_{m}$, see the discussion below.) 
As we will see below  $E(y)$ as well as $E_m$ are isolated eigenvalues of 
finite multiplicities. 
Also}
%

The hamiltonian \eqref{Hy} is called the Born-Oppenheimer hamiltonian. It arises as the central object in the key technique in solving the eigenvalue problem for $H_{\rm mol}$, which is called the Born-Oppenheimer approximation, and which plays an important role in quantum chemistry (for example minima of its ground state energy $E(y)$ 
 determine shapes of molecules). 
 In this approximation the ground state energy  $E(y)$ (or the energy of an excited
state) of $H_N(y)$ is considered as the potential energy of the
nuclear motion, which leads to the Hamiltonian
\[H_{\textrm{nucl}}:= -\sum_{j=1}^M \frac{1}{2m_j}\Delta_{y_j} + E(y).\]
One expects that due to the fact that the ratio of the electron and
nuclear mass being very small, the eigenvalues of
$H_{\textrm{nucl}}$ give a good approximation to the eigenvalues of
$H_{\rm mol}$.
   For rigorous results on the
Born-Oppenheimer approximation see original articles \cite{CDS,
KMSW, LeL, Hag, PST} and the textbook \cite{GS}.

One can define the interaction energy to any order  in the
electron to nuclei mass ratio (see Appendix \ref{sec:beyond BO}),
but this will produce only correspondingly small corrections to our
results.

It is expected, after van der Waals, that $W(y)$ is a sum of pair
interactions, $W_{ij}$, which are attractive and decay at infinity
as $-|y_i-y_j|^{-6}$.
 More precisely, one expects that 
\begin{equation}\label{vdWlaw}
W(y)=-\sum_{i<j}^{1,M}\frac{e^4
\sigma_{ij}}{|y_i-y_j|^6}+ O\big (\frac{e^4}{R^7}\big), 
 \end{equation}
where $\sigma_{ij}$ are positive constants depending only on the
parameters  of the pair of atoms $i,j$ and  
\begin{equation}\label{Rdef}
R=\min\{|y_i-y_j|: 1 \leq i, j \leq M, i \neq j\},
\end{equation} 
provided that $R$ 
 is large enough. Here and in what follows the remainder $O(\dots)$ signifies the behaviour in $y$ and might depend on $N$ and $M$, see however the remark below. 

 Let $E_{m, n},\ n\le Z_m,$ denote the  ground state energy
of the ion with a nucleus of charge $e Z_m$ and $Z_m-n$ electrons and $E_{m}=E_{m, 0}$, the ground state energy 
of the atom corresponding to the $m$-th nucleus of charge $Z_m$. We formulate a property of many body systems playing an important role below:
\begin{itemize}
\item [(E)]  
\qquad $\sum_{i=1}^M E_{i} < \sum_{i=1}^M  E_{i,n_i}, \quad \forall (n_1, \dots n_M): \sum_i n_i=0,\ \sum_i |n_i|> 0.$    
\end{itemize} 

\smallskip
\noindent We discuss this property below. Here we mention only that, since the ground state energies depend on  
the underlying spaces on which the hamiltonians are defined,  Property (E)
 depends on the symmetry type $\s$, defined below, 
and that experiments and numerical computations (see below) show that it holds for all elements for which it was tested, however, theoretically, it proven only for a system of several hydrogen atoms.

\medskip
The simplest symmetry type for fermions of spin $\frac12$ is the one
corresponding to  totally antisymmetric
functions, or, what is the same, to the greatest possible spin.   
We begin with formulating our restlts in this special case.
\begin{theorem}[van der Waals law; highest spin]\label{thm:vdW-maxspin} Assume that the Hamiltonian $H_N(y)$
acts on the space $\mathcal{H}_A=\bigwedge_1^N L^2(\mathbb{R}^3) $ of purely antisymmetric functions
and assume Condition (D) stated below. Then for large distances
between atoms \eqref{vdWlaw} holds for some constants
$\sigma_{ij}>0$ depending on the nature of the atoms $i,j$, if and
only if Property (E) stated above holds.
\end{theorem}
The theorem  above is a corollary of Theorem \ref{thm:vdW-sig},
stated below, dealing with general symmetry types.


Now we define the physical state space, $\cH_{\textrm{fermi}}$, of
the Born-Oppenheimer molecule.  Since electrons are identical
particles and are fermions of spin $\frac{1}{2}$, the state space of
the system of $N$ electrons is the space
\[\bigwedge_1^N(L^2(\mathbb{R}^3) \otimes \mathbb{C}^{2})\]
of  $L^2-$functions,  $\Psi(x_1, s_1, \dots, x_N, s_N)$, of
co-ordinates,  $x_1, \dots, x_N$, and spins, $ s_1, \dots,  s_N$
(with $s_j \in \{-\frac{1}{2},\frac{1}{2}\}, j=1,...,N$) that are
antisymmetric with respect to permutations of pairs $(x_i,s_i)$. The
space $\cH_{\textrm{fermi}}$ is the subspace of $L^2(\R^{3N})$ given
by the projection of the space $\bigwedge_1^N(L^2(\mathbb{R}^3)
\otimes \mathbb{C}^{2})$ onto the $L^2-$functions of the
co-ordinates alone,
\begin{equation*}
\cH_{\textrm{fermi}}:=\{ \lan\chi, \Psi\ran_{\text{spin}} | \Psi\in
\bigwedge_1^N(L^2(\mathbb{R}^3) \otimes \mathbb{C}^{2}),\ \chi:
\{-\frac{1}{2},\frac{1}{2}\}^N \rightarrow \mathbb{C}\},
\end{equation*}
where  $\lan\chi, \Psi \ran_{\text{spin}} :=\sum_{s_1, \dots, s_N
\in \{-\frac{1}{2},\frac{1}{2}\}} \bar\chi(s_1, \dots,  s_N)
\Psi(x_1, s_1, \dots, x_N, s_N)$.

We relate this space to  irreducible representations, $T^\sigma_{S_N}$, of
the group $S_N$, of permutations of $N$ indices. Consider  the
unitary representation $T_{S_N}: S_N \rightarrow U(L^2(\mathbb{R}^{3N}))$
(unitary operators on $L^2$)  of  $S_N$ on the space $L^2(\R^{3N})$, 
 given by $ \pi   \rightarrow T_\pi$, 
with 
\begin{equation}\label{Tpi}
(T_{\pi} \Psi)(x_1,...,x_N)=
\Psi(x_{\pi^{-1}(1)},x_{\pi^{-1}(2)},...,x_{\pi^{-1}(N)}).
\end{equation}
Then the space $\cH_{\textrm{fermi}}$ can be written as
\begin{equation}\label{HFermiHsigma}
\cH_{\textrm{fermi}}=\sum_{\sigma} \mathcal{H}^\sigma,
\end{equation}
where $\sigma$ runs over  irreducible representations of the group
$S_N$ corresponding to at most two-column Young diagrams and
$\mathcal{H}^{\sigma}$ is the subspace of $L^2(\R^{3N})$ on which
this representation reduces to
  multiple of the irreducible representation of the type
$\sigma$ (see Section \ref{sec:setupstat} for definitions and
details). We call  irreducible representation labels $\sigma$ the
\textit{symmetry types}.
\DETAILS{The orthogonal projection on the subspace
$\mathcal{H}^{\lambda}$ is given by
\begin{equation*}
P^{\lambda} =\frac{\chi_{id}^{\lambda}}{N!} \sum_{g \in S_N}
\chi_{g^{-1}}^{\lambda}T_g,
\end{equation*}
with $\chi_g^{\lambda}$ the character of the irreducible
representation $\lambda$ evaluated at $g$.}

 The fact that the electrons  are
identical particles is expressed in the property that $H_N(y)$
commutes with the permutations
\begin{equation}\label{Hsym}
H_N(y)  T_\pi=T_\pi H_N(y), \quad \text{ } \forall \pi \in S_N,
\end{equation}
and therefore the subspaces $\mathcal{H}^{\sigma}$  are invariant
under $H_N(y)$. This allows us to introduce  the ground state energy
of the system for the symmetry type $\s$ by
\begin{equation*}
E^\s(y)= \inf \sigma(H_N(y)|_{\mathcal{H}^\s}).
\end{equation*}

To define $E^\s(\infty)$ we need some notation. Let $a$ be the
decomposition of the molecule into neutral atoms and/or ions,
$S(a)\subset S_N$ be the subgroup of $S_N$ consisting of the
permutations that keep the clusters of $a$ invariant and let $H_a$
be the sum of the corresponding atomic or ionic Hamiltonians (see
the next section for precise definitions). Fix an irreducible representation 
$\sigma=\s (S_N)$ of the group $S_N$. The space $\cH^\sigma$ is invariant under 
 the representation of $S(a)$, but the restriction of  the latter to $\cH^\sigma$ is not necessarily multiple of the irreducible one. Therefore, there exists a family, $I^\s$, of irreducible representations of $S(a)$ such that
 $\cH^\sigma=\oplus_{\alpha  \prec \sigma} \cH_a^\alpha$ (with $\cH_a^\alpha$ non-empty) and  the representation of $S(a)$ on $\cH_a^\alpha$ is multiple of the irreducible $\al$.  
\DETAILS{ Then the restriction $\sigma|_{S(a)}$ (in the sense of the corresponding subspaces) is not necessarily irreducible. Therefore, there
exists a family, $I^\s$, of irreducible representations of $S(a)$
such that
\begin{equation}\label{sig-deco}
\sigma|_{S(a)}=\oplus_{\alpha \in I^\s} {\alpha}.
\end{equation}}

The representations $\alpha = \al (S(a))\in I^\s$ are called
induced representations and we write $\alpha \prec \sigma$. 
(For hydrogen atoms the group $S(a)$ is trivial and the
\textit{construction of induced representations should be omitted}.)
 This definition  implies that the lowest energy of the infinitely separated atoms or ions, when the total system has a symmetry type $\s$, is \begin{equation*}E^{\s}(\infty)=\min_{a,\al\prec \s} \inf \sigma(H_a^\al),\end{equation*}
  where $H_a^\al$ denotes the restriction of $H_a$ onto the subspace on which the representation of $S(a)$ is multiple of the irreducible representation of type $\al$. (Property (E) implies that the minimum can be taken over only atomic decompositions $a$.)
The interaction energy for the symmetry type $\s$ is now defined as
\begin{equation*}
W^\s(y):=E^\s(y)-E^\s(\infty).
\end{equation*}

Finally, we state Condition (D) of Theorem \ref{thm:vdW-maxspin} 
and  of the theorems below. 
 We write
 \begin{equation*}
 \alpha \prec \prec \sigma\ \mbox{if}\ \alpha \prec \sigma\ \mbox{and}\
\inf\s(H_a^{\alpha})= \min_{\beta} \inf \s(H_a^\beta),
\end{equation*} 

\begin{itemize}
\item[(D)] For each atomic decomposition $a$ and for each induced symmetry type $\alpha \prec \prec \sigma$, the
ground state subspace of $H_a^\alpha$ consists only of one copy of
the irreducible representation of type $\alpha$.
\end{itemize}
One expects that for every symmetry type $\alpha$, the ground state
subspace consists of a single copy of the irreducible representation
of the symmetry group, but proving this is an open problem. For a
system of several hydrogen atoms Condition (D) follows from the fact
that $S(a)$ is the trivial group    and from Perron-Frobenious
argument  (see below).

Condition (D) is omitted if the statistics is not taken into account.

Now, we define what we mean by the van der Waals law for fixed
symmetry types and present a result establishing this law.
\begin{definition}[van der Waals - London law for a fixed symmetry type]\label{def:vdW-sig}
We say that the van der Waals law holds for a symmetry type $\s$ if
there exist positive constants $\sigma^{\s,\alpha}_{ij},\ \alpha
\prec \prec \sigma,$ (defined in \eqref{sijsalpha}) such that
\begin{equation}\label{Wsig}W^\s(y)=\min_{\alpha \prec \prec \sigma} W^{\s,\alpha}(y)+O\big (\frac{e^4}{R^7}\big), 
\end{equation} 
where
\begin{equation}\label{Wsig-al}
W^{\s,\alpha}(y):=-\sum_{i<j}^{1,M} \frac{e^4 \sigma^{\s,\alpha}_{ij}}{|y_i-y_j|^6}.
\end{equation}
\end{definition}

The following theorem gives the van der Waals law for fixed symmetry
types.
\begin{theorem}[van der Waals forces for a fixed symmetry type]\label{thm:vdW-sig}
Assume Condition (D) below. Then for every symmetry type $\s$,  the
van der Waals law holds if and only if Property (E) reinterpreted
for the symmetry type $\sigma$ holds.
\end{theorem}
 For a collection of hydrogen atoms Condition  (D) and Property
 (E) are shown below to hold and therefore the van der Waals law is always valid for such a system.

 When $\sigma$ corresponds to a Young diagram with one column (completely anti-symmetric representation) Theorem \ref{thm:vdW-sig} gives Theorem
\ref{thm:vdW-maxspin}. For $\sigma$ corresponding to a Young diagram
of at most two columns the ground state energy of $H_N(y)$ on
$\mathcal{H}_{\text{fermi}}$ is
\begin{equation*}
E(y)=\min_\sigma E^\s(y).
\end{equation*}
  Let $\sigma_0$ be a Yonge diagram of at
most two columns for which $E(y)=E^{\sigma_0}(y)$.
 Then Theorem \ref{thm:vdW-sig} for the
specified $\sigma_0$ gives the interaction energy of the system. (Of
course, if the interatomic distances are not very large it might
happen that the energy surfaces for different symmetries cross and
we have to take $W(y)=E^{\s_0}(y)-E^{\s_1}(\infty)$ where $\s_1\ne
\s_0$ is the diagram that minimizes $E^\s(\infty)$.)

Theorem \ref{thm:vdW-sig} describes the van der Waals  force at a
pairwise large separation between the atoms. For intermediate
distances, the  van der Waals -- London law is modified due to
overlapping between electron clouds of the  atoms and for small
distances,  the van der Waals forces are  repulsive (the energy is
positive) as follows from the rough estimate
$$H_N(y) \geq -C+\sum_{i<j}^{1,M}\frac{e^2 Z_i Z_j}{|y_i-y_j|},$$ for some
 constant $C$ independent of $y$, implied by the
bound $\frac{e^2 Z_m}{|x_n-y_m|} \leq -\alpha \Delta_{x_n}+\beta,$
valid for any $\alpha>0$ and a corresponding $\beta>0$.  
Often the interaction energy for two atoms ($M=2$) is modeled by the
Lennard-Jones potential
$W_{LJ}(y)=\frac{a}{|y_1-y_2|^{12}}-\frac{b}{|y_1-y_2|^6}$ or by the
Buckingham potential $W_B(y)=e^{-c|y_1-y_2|}-\frac{d}{|y_1-y_2|^6}$
where the constants $a,b,c,d$ are determined experimentally.

  If the molecules have dipole moments, then one expects a third power law
 to be true. This can be proven by the techniques developed in this paper.

Now we address Property (E). The following statements are proven in Appendix \ref{sec:propE}:

\begin{itemize}
\item[(a)]  Property (E) holds for a system of several hydrogen atoms and consequently
so is the van der Waals law.
\item[(b)] Property (E) follows from the following property

 (E') For any two nuclei $i$ and $j$, $i \neq j$, in our system and for any integers $m,n \ge 0$, $l > 0$ satisfying 
  $m+l \leq Z_j$, we have the following energy inequality
\[E_{i,m}+E_{j,-n}< E_{i,m+l}+E_{j,-n-l}.\]

\item[(c)]  Property (E) for a system of hydrogen atoms with electron statistics follows from  Property (E) without statistics. 
\end{itemize} 
The meaning of Property (E') is that ionization energies of  atoms
are greater than the electron affinities, where, in a standard
terminology,  the $n$-th ionization
 energy ($n \geq 1$) of an atom is the energy required to remove an electron from
 its $n-1$-ion and the $n$-th electron affinity is the energy required
 to remove an electron from its $-n$ ion. The table below, taken from \cite{MA}, gives
 first ionization energies and first electron affinities of atoms.  It shows that the first ionization energies are always
 much larger than the first electron affinities.

It is experimentally verified that the higher (second third and so
forth) ionization energies are bigger than the first and it is
expected that the higher electron affinities affinities are
 lower or zero. (Only ions with at most two extra electrons -- i.e. of the charges
at most $-2e$ -- are observed experimentally.

That nuclei can bound
only finite number of electrons was proven in \cite{Rus, Sig} (with
the bound $\le 17Z$ on the number of the extra electrons). It was
shown in \cite{LSST}, with some improvements by \cite{SSS, FS1},
that asymptotically ions are neutral and in \cite{Lieb} that the
maximal number of extra electrons $\le Z+1$. The latter bound was
recently improved to $<0.22\/Z+3Z^{1/3}$  in \cite{Phan}. Moreover,
it was shown in \cite{Sol1, Sol2} that in the Hartree - Fock
approximation the maximal number of extra electrons $\le
\text{const}$.)

 Our approach allows for reasonable estimates of the dependance of various remainders $O(\dots)$ on $N$ and $M$, which leads to an estimate of allowed internuclear distance. 
(We note that the Lieb-Thiring upper bound in   \cite{LT} on the interaction energy  has a fairly good control of the allowed internuclear distance.) 

\paragraph{Remark.}  Using some ideas of this paper, the first author obtained a simple and well-behaved in $N$ upper bound, not requiring Condition (E) and proved lower bounds with the reasonable $N-$dependence of the remainder estimates (\cite{Anap}). 

\bigskip

 \begin{tabular}{|c|c|c|c|}
\hline Atomic number & Element & first Ionization & Electron affinity  \\
  &  &  energy (kcal/mol) & (kcal/mol) \\
\hline 1 &  H  & 313.5 &  17.3\\
\hline 2 &  He  & 566.9 &  - \\
\hline 3 &  Li  &  124.3 &  (14) \\
\hline 4 &  Be  & 214.9 &  - \\
\hline 5 &  B  &  191.3 &  (7) \\
\hline 6 &  C  &  259.6 &  29 \\
\hline 7 &  N  &  335.1 &  - \\
\hline 8 &  O  &  314.0 &  34 \\
\hline 9 &  F  &  401.8 &  79.5 \\
\hline 10 & Ne  & 497.2 &  -  \\
\hline 11 & Na  & 118.5 &  (19)  \\
\hline 12 & Mg  & 176.3 &   - \\
\hline 13 & Al  & 138.0 &   (12) \\
\hline 14 & Si  & 187.9 &  (32)  \\
\hline 15 & P  &  254 &   (17) \\
\hline 16 & S  &  238.9 &  47\\
\hline 17 & Cl  & 300.0 &  83.4  \\
\hline 18 & Ar  & 363.4 &  (16) \\
\hline 19 &  K &  100.1 & - \\
\hline 20 & Ca  & 140.9 &  - \\
\hline 21 & Sc  & 151.3 &  - \\
\hline 22 & Ti  & 158 &  - \\
\hline 23 &  V &  155 &  - \\
\hline 24 &  Cr & 156  &  - \\
\hline 25 &  Mn &  171.4&  - \\
\hline
\end{tabular}
\medskip
\newline
(Values in parentheses are estimated by quantum-mechanical
calculation and have not been verified experimentally.)

\paragraph{Outline of the approach.}  Our approach is based on the Feshbach-Schur   perturbation argument, with the small parameter being the reciprocal of in the 
distance between the nuclei. 
  In what follows we omit the argument $y$ and write $E$ and $H$ for $E(y)$ and $H_N(y)$, respectively. 
  
\noindent {\it Feshbach-Schur  method.}  This method originates in the works of H. Feshbach and I. Schur 
and was reformulated and generalized in  \cite{BFS}. We follow the textbook presentation of \cite{GS}. Let $P$ be an orthogonal projection and $P^\bot=1-P$. Introduce the notation  $\H^\bot=P^\bot \H P^\bot$.
Assume
\begin{itemize}
\item[(a)]  $\Ran(P) \subset D(\H)$ (domain of $\H$) and therefore $\|HP\|< \infty$;
\item[(b)]   The operator $(\H^\bot-\lambda)$   is invertible.
\end{itemize}
 The Feshbach-Schur method, 
 as applied to the quantum Hamiltonian $\H$, states that if Conditions (a) and (b) are satisfied, then the Feshbach-Schur map
\begin{equation}\label{FP}
F_P(\lambda)=(P \H P-U(\lambda))|_{\Ran P},
\end{equation}
where
\begin{equation}\label{U}
U(\lambda):=P \H P^\bot (\H^\bot-\lambda)^{-1} P^\bot \H P,
\end{equation}
is well defined and
\begin{equation}\label{FSE}
\lambda \text{ eigenvalue of } \H \iff \lambda\ \text{ eigenvalue of
}\ F_P(\lambda).
\end{equation}
Moreover, the eigenfunctions of $\H$ and $F_P(\lambda)$
corresponding to the eigenvalue $\lam$ are connected as
\begin{equation}\label{Fesh-iso}
H \psi=\lambda \psi\ \quad \Leftrightarrow\ \quad  F_P(\lam)
\phi=\lambda \phi,
\end{equation}
where $\phi,\psi$ are related by the following equations
$\phi=P\psi, \text{ } \psi=Q(\lambda) \phi. $ 
Here the family of operators $Q(\lam)$ is defined  as
$Q(\lam)=P-P^\bot (H^\bot-\lambda)^{-1} P^\bot H P. $ 
(Remember that we do not display the $y-$dependence of various
objects.)

\bigskip

\medskip

\noindent {\it Orthogonal projection $P$.} 
To apply the Feshbach - Schur map, we have to choose the orthogonal projection $P$. 
We fix an irreducible representation $\s$ of $S_N$ and let  $H^{\s}$ be the restriction of $H$ to the subspace of the irreducible representation $\s$. We denote by $\mathcal{A}^{at}$ the collection of all decompositions $a=(A_1,...,A_M)$, with $|A_j|=Z_j$ for all $j=1,...,M$. (Its elements correspond to decompositions of our system to neutral atoms.) 
  Let  $V_{a, R}^\al$ be the ground state subspace  of 
 $H_a^\al$ (the corresponding energies  define $E^\s(\infty)$),  
cut-off at large distances, so that $V_{a, R}^\al$ for various $a \in \mathcal{A}^{at},\ \al \prec \prec \sigma,$ become mutually orthogonal. Let $P$ be  the orthogonal projection on span $\{V_{a, R}^\al,\ a \in \mathcal{A}^{at},\ \al \prec \prec \sigma\}.$  (Here $R$ signifies the scale on which we perform the cut-off.) 
As a result of cutting-off the ground state subspaces  of  $H_a^\al$, we have the representation
\begin{equation}\label{P-deco'}
P=\sum_{a \in \mathcal{A}^{at},\ \al \prec \prec \sigma} P^\al_{a, R},
\end{equation}
where $P^\al_{a, R}$ are the orthogonal projections onto $V_{a, R}^\al$. 

 To show that the Feshbach - Schur map exists, we note that $\Ran (P) \subset
\text{Dom}(H)$ 
and therefore the condition (a) for the existence of the Feshbach map holds. We will prove in
Section \ref{Hbotbndseveral} that, under Property (E), there is a
$\imsgap > 0$, independent of $y$ s.t. for $R$ large enough, the following \textit{stability bound}  holds:
\begin{equation}\label{Hbotbnd'}
\H^{\s \bot} \geq E^\s(\infty)+2\imsgap,
\end{equation}
where  $H^{\s \bot}=P^{\bot} H^\s P^{\bot}$ and $E^\s(\infty)$ is the lowest possible energy for all possible break-ups of the system,   defined in \eqref{Einftydef} and $\imsgap=\g^\s$ is a positive number independent of $N$ and $M$.  Thus the condition (b) is also
satisfied and, for all $\lambda \leq E^{\s} (\infty)+\gamma$, the Feshbach - Schur map $F_P(\lambda)$ is well defined.
\DETAILS{ for our approach
without taking the statistics into account. Our goal is to prove Theorem \ref{thm:vdWlaw-nostat},
 under certain technical assumptions, which are then verified in later sections.

, recall,  the notation  $\H^\bot=P^\bot \H P^\bot$. We
\textit{assume} this for now. Using \eqref{Hbotbnd}
 we obtain that $H^\bot-\lambda$ is invertible for all
$\lambda \leq E(\infty)+ \imsgap $.

This argument is based on the \textit{stability bound},
\begin{equation}\label{Hbotbnd'}
\H^\bot \geq \Einfty+2\imsgap,
\end{equation}
where $\Einfty$ and  $\H^\bot=P^\bot \H_N(y) P^\bot$, with $P^\bot=1-P$ and  $P$, the orthogonal projection onto the span of the ground states of the break-ups with the energy $\Einfty$.}
%
  Now, to use \eqref{Fesh-iso} we have to estimate the two terms on the r.h.s of \eqref{FP}.

\medskip

\noindent {\it Estimate of $P \H^{\s} P$.} 
 First,  we use the equations \eqref{P-deco'}, the decomposition $H=H_a+I_a$, where $I_a$ is the sum of the intercluster interactions in the decomposition $a$,  and  $H_a P_{a, R}^\alpha \doteq \Einfty P_{a, R}^\alpha$, where the dot above the relations stands for exponentially small additive terms omitted, to obtain $ P H^{\s} P \doteq \sum_{a, b,\al, \beta}  P_{a, R}^{\al} I_b P_{b, R}^{\beta}.$ Since the supports of $V^\al_{a, R},$ are mutually disjoint, this gives $ P H^{\s} P \doteq \sum_{a,\al}  P_{a, R}^{\al} I_a P_{a, R}^{\al}.$
Clearly, the map $P_{a, R}^\alpha I_aP_{a, R}^{\alpha}\big|_{\Ran P_{a, R}^{\alpha}}$ leaves the space $V_{a, R}^{\al}=\Ran P_{a, R}^{\al}$ invariant and  commutes with $T_{\pi},\ \forall \pi\in S(a)$.  Since by Condition (D) $\Ran P_{a, R}^{\alpha}$ is a space of an  irreducible  representation of $S(a)$, we conclude that  it is a multiple of the identity, 
$$P_{a, R}^\alpha I_a P_{a, R}^{\alpha}\big|_{\Ran P_{a, R}^{\alpha}}=
\frac{1}{\rank(P_{a, R}^{\alpha})}\Tr (I_a P_{a, R}^{\alpha}) P_{a, R}^{\alpha},$$
 where $\rank(P_{a, R}^{\alpha})$ is the rank of $P_{a, R}^{\alpha}$. 
Hence $P_{a, R}^{\al} I_a P_{a, R}^{\al}$ can be written 
as 
 \begin{equation}\label{TrIP}\Tr (I_a P_{a, R}) = e^2\sum_{i<j}^{1,M} 
  Z_i Z_j I(z, z'; y_{ij})\rho_{A_i, R}(z)\rho_{A_j, R}(z') dz dz' ,\end{equation}
where  $\rho_{A_j}^\al$ are   the  one-electron densities  of atoms $A_j$ in the irreducible representation $\al$,  defined in \eqref{rhoA-def} below, $y_{ij}=y_i-y_j$,  and, by  the charge neutrality for each atom $A_j$, 
\begin{equation}\label{tildeIijkl}
I ( z, z'; y)=-\frac{1}{|y+z|}-\frac{1}{|y-z'|}+\frac{1}{|y|}+\frac{1}{|y+z-z'|}.
\end{equation}
Since as shown below, these densities   are spherically symmetric,  Newton's screening theorem gives  
$$P \H^{\s} P = \Einfty P +\ \mbox{exponentially small terms}.$$ 
\DETAILS{, and using the charge neutrality for atoms $A_j$ and Newton's screening theorem, we show that 
$$P \H^{\s} P = \Einfty P +\ \mbox{exponentially small terms}.$$ 
The key point here is that to apply this theorem, we have show that  
  $\rho_{A_j}^\al$  are spherically symmetric. 
The latter fact follows from Condition (D) and an elementary analysis.} 
\DETAILS{Consider a hamiltonian $H_{A}$ of an atom or ion (i.e. $H_{A}$ is of the form \eqref{Ha1} below) and let $H_{A}^\al$ be this hamiltonian restricted to the subspace of  the symmetry type $\al$. Let, furthermore,  $P_{A, R}^\al$ be the orthogonal projection onto  the ground state subspace of $H_{A}^\al$ and $\rho_{A}^\al$,  the  one-electron density for $H_{A}^\al$,  defined, say, through the trace relation $\Tr (b \rho_{A}^\al)= \Tr [(b\otimes \one) P_{A, R}^\al]$ for any one-electron operator $b$. ( can be written explicitly in terms of  any orthonormal basis in $\Ran P_{A, R}^\al$, see Appendix \ref{sec:sym-gen-app}.)
Let $\rho_{A}^\al$ be  the one-electron density $\rho_{A}^\al$ of an atom $A$ in the irreducible representation $\al$,  defined by the  condition \begin{equation}\label{rhoA-def} \Tr (b \rho_{A}^\al)= \Tr [(b\otimes \one) P_{A}^\al],\end{equation} 
where $P_A^\al$ is the orthogonal projection onto the ground state space of the atom hamiltonian $H_A^\al$, at the symmetry type $\al$. } 
The above implies that the van der Waals decay law could only come from $U(\lambda)$. 
\DETAILS{To estimate  $U(\lambda)$, we need, in additions to \eqref{Hbotbnd'}, similar  estimates for the boosted hamiltonians $e^{-f( x)} \H^{\s \bot} e^{f ( x)}$, with various weights $f(x)$. These bounds allow us to pull exponential weights coming from eigenfunctions $\Phi_a\in V_a$ through the resolvent $ (\H^{\s \bot}-\lambda)^{-1} $ and use that a product of the weights coming from functions like  $\Phi_a\in V_a^\al$ and $ \Phi_b\in V_b^\beta,\ a\ne b,$ decays exponentially in all directions. 
(Alternatively, we also use a geometrical decomposition of the operator $H^{\s \bot}$.)}

\medskip

\noindent {\it Estimate of $U(\lambda)$.}  To estimate $U(\lambda)$, we use the equations \eqref{P-deco'}, 
 $H=H_a+I_a$ and  $H_a P_{a, R}^\alpha \doteq \Einfty P_{a, R}^\alpha$ again and use $P^\bot P_{a, R}^{\al}=0$, to obtain 
$ P^\bot H P \doteq \sum_{a,\al}  P^\bot I_a P_{a, R}^{\al}.$
 Now, using this in the definition of  $U(\lambda)$,   we obtain  
  \begin{align}\label{Usig-aprox1'}
U^\s (E) \dot = \sum_{a,b,\alpha, \beta} 
 P_{a, R}^{\al}  I_a   P^\bot R^{\bot}(E) P^\bot  I_b P_{b, R}^{\beta}. 
\end{align}
Next,  we  use the variables $z_{k m}=x_k-y_m,\ \forall  k \in A_m,\  m=1,...,M,$  
and expand the intercluster interaction $I_{a}$  in $|y_{ij}|^{-1}$ and use properties of  $P_{a, R}^{\al}$ and  the charge neutrality for each atom $A_j$,  to obtain 
  \begin{equation}\label{IaPa-exp'} 
I_{a}P_{a, R}^{\al}= \sum_{i<j} [\frac{e^2}{|y_{ij}|^3}f_{ij}(z,\widehat{y_{ij}})P_{a, R}^{\al}+ O(\frac{e^2}{|y_{ij}|^4})], \end{equation}
where $f_{ij}(z,\widehat{y_{ij}})=\sum_{k \in A_i, l \in A_j}  [z_{ki} \cdot z_{lj}- 3(z_{ki} \cdot \widehat{y_{ij}})(z_{lj} \cdot  \widehat{y_{ij}})]$. This gives, in particular, $\|I_a P_{a, R}^{\al}\| \lesssim \sum_{i<j} \frac{1}{|y_{ij}|^3}$. Using  this bound  and 
 elementary geometrical localization arguments, which allow us  to pass from  $ (H^{\s \bot}-E )^{-1}$ 
to $(H^{\al \bot}_a -E )^{-1}$, where $H^{\al \bot}_a:=H^{\al }_a P^{\al \bot}_a$, with $P^{\al \bot}_a$ the projection onto the orthogonal complement of  the ground state subspace  of  $H_a^\al$, and to estimate the terms with $a\ne b$, and  pass from  $E$  to $E(\infty)$, we find
  \begin{align}\label{FPsigmaaprox'}
U (E) = \sum_{a,\alpha, \beta} 
 Q U_{aa}^{\alpha \beta}  Q+ O\big (\frac{e^4}{R^7}\big), 
\end{align}
where
$ U_{aa}^{\alpha \beta}:= P_{a, R}^{\al}  I_a   P^{\al \bot}_a R^{\al \bot}_a  P^{\al \bot}_a  I_a P_{a, R}^{\beta},$ with 
$R^{\al \bot}_a :=  (H^{\al \bot}_a -E(\infty) )^{-1}$. 
 The 
 orthogonality of different irreducible representations, that the terms with $\al\ne \beta$ vanish. Finally, since $U_{aa}^{\alpha \beta}\big |_{\Ran P_{a, R}^{\al}}$
acts on the space of an  irreducible representation of $S(a)$ and commutes with all the operators of this representation, 
it is multiple of identity. 
This implies
 \begin{align}\label{Uaaalal'}U_{aa}^{\alpha \beta} = \frac{1}{\rank P_{a, R}^{\al}}\Tr (U_{aa}^{\alpha \alpha} P_{a, R}^{\al}) P_{a, R}^{\al}\del_{\al, \beta}. \end{align}
The terms $\Tr (U_{aa}^{\alpha \alpha} P_{a, R}^{\al})$ 
 can be easily computed, using \eqref{IaPa-exp'}  and 
 showing  that the terms with $(ij)\neq (i'j')$ or $(ij)= (i'j')$, but $ (k l)\neq (k' l')$, coming from different $I_a$'s in the expression for $\Tr (U_{aa}^{\alpha \alpha} P_{a, R}^{\al})$ 
 can be written as integrals of odd function and therefore vanish. This gives 
 $$ U_{aa}^{\alpha \al} = \frac{1}{\rank P_{a, R}^{\al}}\sum_{i<j} \frac{e^4}{|y_{ij}|^6}\Tr (f_{ij} R^{\al \bot}_a f_{ij} P_{a, R}^{\al})+O(\frac{e^4}{R^7}),$$ which yields the van der Waals decay law. (The terms in the sum can be further simplified.)

\paragraph{Remark.}  
Due to \eqref{IaPa-exp'}, for the charge neutral systems, the term $P \H^{\s} P$  always starts, at least, with 
 $\sum_{i<j} |y_i-y_j|^{-3}$. 

\paragraph{Beyond  Born-Oppenheimer Approximation.}\label{sec:beyond BO} 
Let $\psi_{\rm BO}(x, y)$ be the ground state of $H_N(y)$,
normalized as $\int|\psi_{\rm BO}(x, y)|^2d x=1$. Using  the
Feshbach - Schur map with the orthogonal projection


$$ (P f)(x, y) = \psi_{\rm BO}(x, y)\int \overline{\psi_{\rm BO}(x, y)}f(x, y) dx $$
(i.e. integrating out the electronic degrees of freedom),  as it is done e.g. in \cite{GS}, we can show that 

 \[ \lam \in \sigma_d(H_{{\rm mol}}) \; \leftrightarrow \;
  \lam \in \sigma_d (H_{\textrm{nucl}} (\lam)),\]
 with the corresponding eigenfunctions related accordingly. Here $H_{\text{nucl}} (\lam)$ is the operator on $L^2(\R^{3M})$, defined by

 \smallskip

 $$H_{\text{nucl}} (\lam) :=-\sum_{j=1}^M \frac{1}{2 m_j} \Delta_{y_j}+E_\kappa(y, \lam) ,$$
where $\kappa:=m/\min_j m_j$ and

 $$ E_\kappa(y, \lam) = E(y)+
 \sum_1^M \frac{1}{2m_j}\int |\nabla_{y_j}\psi_{BO}|^2 dx 
 +O(\kappa^2),$$
with $O(\kappa^2)$ standing for a non-local operator of the
indicated order. The energy $E_\kappa(y)$ (not potential anymore)
can be used to define the interaction energy in all orders.

\medskip

\paragraph{Organization of the paper.} 
In Section \ref{sec:prelim} we discuss preliminaries of quantum many body systems. In Section \ref{sec:setup} we prove Theorem \ref{thm:vdW-maxspin} 
modulo the stability estimate \eqref{Hbotbnd'} which is proven in Section \ref{Hbotbndseveral}. 
In Section \ref{sec:proofvdWThm-s}
 we rework the proof of Theorem \ref{thm:vdW-maxspin} to prove Theorem \ref{thm:vdW-sig}. 

In Supplement I, we present an estimates for the boosted hamiltonians $e^{-f( x)} \H^{\s \bot} e^{f ( x)}$, with various weights $f(x)$,  similar  to \eqref{Hbotbnd'}. These bounds allow us to pull exponential weights through the resolvents. They are not used in this paper, but could be useful in various extensions (see e.g. \cite{Anap}). 

\paragraph{Notation.} We collect here general notation used in this
paper. In the text $C$ will denote a positive constant which might
be different from one equation to the other, but which is
independent of the nuclear co-ordinates $y_1, \dots y_M$. We will
use the notation $\lesssim$ for inequalities that are true up to
such a constant. 
We will write $A\ \doteq\ B$, and  $A\ \dotle\ B$,  $A\ \dotge\ B$  if $\|A-B\|\ \lesssim\ e^{-\frac16 \theta R}$, $A \le B + C e^{-\frac16 \theta R}$ and $A\ge B- C e^{-\frac16 \theta R}$, for some $C>0$, respectively, in an
appropriate norm, where $\theta$ is the constant appearing in \eqref{groundstatedecay} and \eqref{Phiadecay}. 

For any Banach space $X$, we denote $B(X):=\{f:X \rightarrow X: f$
linear and bounded$\}$. For an operator $A$, the symbols $\s(A)$ and
$\s_{\textrm{ess}}(A)$ stand for the spectrum and the essential
spectrum, correspondingly. We also use the notation $\langle x
\rangle=(1+|x|^2)^{\frac{1}{2}}$. 
Finally, $\|\cdot\|$ will denote the $L^2 - $
 norm of a function or the $B(L^2) - $
norm of an operator, depending on the context,  and the
symbol $O(\delta)$ is understood in this norm.  For an orthogonal projection $P$, we denote by  $P^\bot$, the complementary projection $P^\bot:= \one - P$  and for a vector $\Phi$, we denote by   $P_{\Phi}$ the rank-one  orthogonal projection 
 onto $\Phi$. 

\medskip
\noindent {\bf Acknowledgements.} 
The first author is grateful to Volker Bach, Alessandro
Giuliani and Artem Dudko for useful discussions, and especially to
Marcel Griesemer for numerous inspiring discussions. His research
was supported in part by NSERC under Grant No. NA7901 and in part by
the German Science Foundation under Grant No. 3213/1-1. 
 The second author is grateful to J{\'e}r{\'e}my
Faupin, Alessandro Pizzo and, especially J\"urg Fr\"ohlich and
Elliott Lieb for very useful and inspiring discussions. His interest in the problem of the van der Waals forces derived from  several talks by and conversations with Herbert Spohn on the Casimir- Polder effect. His research
was supported in part by NSERC under Grant No. NA7901.


\bigskip

\section{Preliminaries about many body systems} \label{sec:prelim} 

\paragraph{General properties of $H_N(y)$.}
The general information on the spectrum of the Hamiltonian
\eqref{Hy} is given in the following theorem which is a special case
of the HVZ Theorem (see e.g. \cite{HS, GS, CFKS}).
 \begin{theorem}[HVZ theorem]\label{hvz}
$\sigma_{\textrm{ess}}(H_N(y))=[\Sigma,\infty),$ where $\Sigma=\inf
\sigma(H_{N-1}(y))$.
\end{theorem}
This theorem says that the essential (continuous) spectrum of
$H_N(y)$ originates from the molecule shedding of an electron which
moves freely at infinity and therefore whose energy spectrum changes
continuously. The next results, due to G. Zhislin and J.-M. Combes and L. Thomas, respectively, show that $H_N(y)$, as well as each
atom, has a well-localized ground state (see e.g. \cite{HS, GS, CFKS}):
\begin{theorem}[Zhislin theorem]\label{zysl} The operator $H_N(y)$ has infinite number of eigenvalues, $E_j$, below its essential spectrum,
  $E_j<\Sigma$. 
\end{theorem}
\begin{theorem}[Combes - Thomas bound]\label{zysl} The eigenfunctions, $\Phi_j$, of $H_N(y)$, corresponding to the eigenvalues, $E_j<\Sigma$,   are exponentially  localized: 
 \begin{equation}\label{exp-bnd}
|\Phi_j(x)|\le Ce^{-\delta|x|},
\end{equation}
for any $\delta<\sqrt{\Sigma-E_j}$. Here $x=(x_1,\dots,x_N)$. 
\end{theorem}
The last two theorems show that the atoms and Born-Oppenheimer molecules are stable in the sense that at
sufficiently low energy they are well localized in the space.
However, it says nothing about stability of true molecules. The
likely source of instability of molecules is not shedding of an
electron but breaking up into atoms or ions which in total have the
same energy as the molecule. So far the only molecule proven to be
stable is the hydrogen molecule $H_2$ (see \cite{BFGR}).

Theorems \ref{hvz} - \ref{zysl} still hold for fixed symmetry
types, with the statement of Theorems \ref{hvz} modified as
\begin{itemize}
\item $\sigma_{\textrm{ess}}(H_N^\s(y))=[\Sigma^\s,\infty),$ where $\Sigma^\s=\min_{\al \prec \prec \s}\inf \sigma(H_{N-1}^\al (y))$.\end{itemize}

  The uniqueness of the ground state is a delicate issue. Without
  statistics the ground state of Schr\"odinger operators $H$ is unique (non-degenerate).
  This follows from the positivity improving property of $e^{-\beta H},
  \beta >0$ and from Perron-Frobenious theory  (see for example  \cite{RSIV}). 
  For spaces with statistics the ground state energy
  is in general degenerate (for an irreducible representation  $\sigma$ of the permutation group $S_N$,  it is at least the dimension of this representation) but its  multiplicity is not known. 
However, 
 \begin{align}\label{dim-gr-state-sym} &\mbox{Under  Condition (D), the dimension of the  ground state subspace for the symmetry type }\s \notag\\ &=\mbox{   the dimension of  irreducible  representation $\s$.}\end{align}
 \DETAILS{it is exactly the dimension of the corresponding irreducible  representation and therefore 
\begin{equation}\label{uniq-gr-state} \mbox{ for  one-dimensional  irreducible  representations, the ground states  are  unique.}\end{equation}}

Our next result concerns properties of the  one-electron densities,  $\rho_{A}^\al$, 
mentioned in the introduction. Consider a hamiltonian $H_{A}$ of an atom or ion (i.e. $H_{A}$ is of the form \eqref{Ha1} below) and let $H_{A}^\al$ be this hamiltonian restricted to the subspace of  the symmetry type $\al$. Let, furthermore,  $P_{A}^\al$ be the orthogonal projection onto  the ground state subspace of $H_{A}^\al$. We define   the  one-electron density, $\rho_{A}^\al$, for $H_{A}^\al$,  say, through the trace relation  by the  condition 
\begin{equation}\label{rhoA-def} 
\Tr (b \rho_{A}^\al)= \Tr [(b\otimes \one) P_{A}^\al],\end{equation} 
for any one-electron operator $b$. ($\rho_{A}^\al$ can be written explicitly in terms of  any orthonormal basis in $\Ran P_{A}^\al$, see Appendix \ref{sec:sym-gen-app}. Also, the one-electron density can be associated to any orthogonal projection on  $L^2(\R^{3N})$.) We have 

\begin{proposition}\label{prop:spherical}
Let $H_{A}$ be a hamiltonian of an atom or ion.  Then the one electron density of the ground state subspace of $H_{A}$ is spherically symmetric. (If $H_{A}$ is defined on the entire space (no symmetry restriction) then the  ground state of $H_{A}$ is spherically symmetric.)
\end{proposition}
\begin{proof} By 
the Riesz formula for eigen-projections, $P_{A}^\al$ and therefore $P_{A, R}^\al$ commutes with   any rotation,  $P_{A, R}^\al T_R=T_R P_{A, R}^\al$, $\forall R\in O(3)$, where 
\begin{align} \label{TR} T_R \Phi(z_1,...,z_{|A|})=\Phi(R^{-1} z_1,...,R^{-1} z_{|A|}), \end{align} 
 we have  
 $\Tr [T_R^{-1} (b\otimes \one) P_{A, R}^\al T_R]= \Tr [(t_R^{-1} b t_R\otimes \one) P_{A, R}^\al]$, where $t_R $ denotes the one-electron rotation. This, together with the definition  $\Tr (b \rho_{A}^\al)= \Tr [(b\otimes \one) P_{A, R}^\al]$ of $\rho_{A}^\al$ and the cyclic property of the trace,  gives   $\Tr (b \rho_{A}^\al)= \Tr [(t_R^{-1} b t_R)  \rho_{A}^\al]= \Tr [ b   (t_R\rho_{A}^\al t_R^{-1})]$ and therefore $ \rho_{A}^\al= t_R\rho_{A}^\al t_R^{-1} $.
 
 With no symmetry restriction, $P_{A, R}^\al$  is a rank-one projection and therefore  $P_{A, R}^\al T_R=T_R P_{A, R}^\al$, $\forall R\in O(3)$,  implies that  the  ground state of $H_{A}$ is spherically symmetric. 
 \end{proof}
 We mention here that the definition of $\rho_{A}^\al$ implies that for any orthonormal basis, $\{\Psi_{A}^{\alpha,i}, i=1,..., n_A\}$ in $\Ran P_{A, R}^\al$, where $n_A:=\dim \Ran P_{A,R}^\al$, we have
\begin{equation}\label{rho1}
\rho_{A}^\al(z_1):=\sum_{i=1}^{n_A} \int |\Psi_{A}^{\alpha,i}(z_1,...,z_{|A|})|^2 dz_2...dz_{|A|}.
\end{equation}

\paragraph{Decompositions.}\label{sec:deco}
Recall that $M$ and $N$ are the numbers of the nuclei and electrons.
Let $a=(A_1,...,A_M)$ be a partition of $\{1,2,...,N\}$ into disjoint subsets some of which might be empty. With the set $A_j$ we
associate the $j$-th nucleus of charge $eZ_j$ by assigning to it the
electrons with labels  in $A_j$. This gives a decomposition of the system into atoms/ions,  $A_1,...,A_M$, called also  clusters. We denote
the collection of all such decompositions by $\mathcal{A}$. 
The set of all $a \in \mathcal{A}$ with $|A_j|=Z_j$ for all $j=1,...,M$ will
be denoted by $\mathcal{A}^{at}$. Its elements correspond to
decompositions of our system to neutral atoms.

Permutations $\pi \in S_N$  act naturally on decompositions $a \in \mathcal{A}^{at}$.  
For any two non-equal decompositions, $a$ and $b$, there is a unique permutation $\pi \in S_N/ S(a)$  such that $b=\pi a$. Here, recall, $S(a)$ is the subgroup of $S_N$, which leaves $a$ invariant.  

 For each  decomposition $
a=(A_1,...,A_M) \in \mathcal{A}$ we define the hamiltonian
\begin{equation}\label{Ha}
H_a=\sum_{m=1}^{M}H_{A_m},
\end{equation}
where  $H_{A_m}$ is the Hamiltonian of the m-th atom or ion,
\begin{equation}\label{Ha1}
 H_{A_m}:=\sum_{i\in A_m}(- \Delta_{x_i}- \frac{e^2Z_m}{|x_{i}-y_{m}|})+\sum_{i,j \in A_m, i<j}\frac{e^2}{|x_{i}-x_{j}|},
\end{equation}
and the inter-cluster interaction
$I_a:=H_N(y)-H_a$, 
the sum of all interactions between the different atoms/ions in the decomposition
$a$. We have that
\begin{equation}\label{Hadecomp}
 H_N(y)=H_a+I_a.
 \end{equation}
 Let  $ E_{A_m}$  and $\phi_{A_m}$  and $E_a$ and $\Phi_a$ be the ground state energy and ground state of  $H_{A_m}$ and $H_a$, respectively,  $H_{A_m} \Phi_{A_m}=E_{A_m} \Phi_{A_m}$ and $H_a \Phi_a=E_a \Phi_a$.
  We have that $E_a=\sum_{m=1}^M E_{A_m}$ and
\begin{equation}\label{Phiaeq}
 \Phi_{a}(x_1,...,x_N)= \prod_{m=1}^M \phi_{A_m}(x_{A_m}),
\end{equation}
where  
$x_{A_m}=(x_i: i \in A_m)$. 
Furthermore, if  $\phi_m$  the ground state of the  $m$-th atom with the nucleus fixed at the origin, then
 \begin{equation}\label{phiAj}
\phi_{A_m}(x_{A_m})=\phi_{m}(x_{A_m}-y_m),
 \end{equation}
with $x_{A_k}-y_m=(x_i-y_m: i \in A_m)$. 
 Throughout the text we will always assume that $\|\phi_m\|=1$ for
all $m=1,...,M$. Standard estimates (see \cite{HS}) give
 that there exists $\theta >0$ such that
\begin{equation}\label{groundstatedecay}
\|e^{\theta \langle x_{A_m} \rangle} \partial^{\alpha} \phi_{m}\|
\lesssim 1, \forall \alpha \text{ with } 0\leq |\alpha| \leq 2,
\end{equation}
where $\alpha$ is a multiindex with each index corresponding to
differentiation in some body variable. For $a =\{A_1,...,A_M\}\in
\mathcal{A}$ we define $y_{a}=(y_{a,1},...,y_{a,N})$, where
$y_{a,i}=y_m$ for $i \in A_m$. In other words $y_{a,i}$ is the
coordinate of the nucleus that $x_i$ is assigned to. From
\eqref{Phiaeq}, \eqref{phiAj} and \eqref{groundstatedecay} we obtain
that
\begin{equation}\label{Phiadecay}
\|e^{\theta \langle x-y_a \rangle} \partial^{\alpha} \Phi_a\|
\lesssim 1, \forall \alpha \text{ with } 0 \leq |\alpha| \leq 2,
\end{equation}
where $x=(x_1,...,x_N)$.

 By the definition of $a$'s, $\min_{a \in \mathcal{A}} E_a=\Einfty$, where $\Einfty$ was defined in
\eqref{Einftydef}. If Property (E) holds, as we expect it always does, than $ E_a=\Einfty$, $\forall a \in \mathcal{A}^{at}$.

\DETAILS{ Then  Property (E) implies that 
\begin{equation*}
E^{\s}(\infty)=\min_{a\in \cA^{at},\al\prec \prec \s} \inf \sigma(H_a^\al).
\end{equation*}}

\bigskip


\section{Proof of Theorem \ref{thm:vdW-maxspin} assuming stability estimates} \label{sec:setup} 
\DETAILS{In this section we lay out the general set-up for our approach without taking the statistics into account. Our goal is to prove
Theorem \ref{thm:vdWlaw-nostat},}
In this section we prove Theorem \ref{thm:vdW-maxspin} stating van der Waals law in the  highest spin case, 
 under a technical assumption, which are then verified in later sections. In what follows we omit
the subindex $N$ and the argument $y$ and write $E$ and $H$ for $E(y)$ and $H_N(y)$. 

For the system to have the  highest spin, $H$ has to act on the subspace  $\mathcal{H}_A=\bigwedge_1^N L^2(\mathbb{R}^3) $ corresponding to  one-dimensional, anti-symmetric representation of $S_N$ (with the one-column Young diagram).
For each $a$, this  representation induces the unique,  the one-dimensional, anti-symmetric representation of $S(a)$.  Hence we omit, without a danger of confusion, the label for this representation, having in mind that all operators below act on  the subspace $\mathcal{H}_A=\bigwedge_1^N L^2(\mathbb{R}^3) $ 
and, as can be easily verified leave this subspace invariant. 
\paragraph{Orthogonal projection $P$.} \label{FPexistence}
 We will now define the projection $P$, to be used in the Feshbach-Schur method described above.
  We cut off the ground state energy $\Phi_a$ defined in
\eqref{Phiaeq} as follows. Recall that $\phi_j$ denotes the ground
state of the $j$-th atom centered at the origin. Let $\chi_R:
\mathbb{R}^3 \rightarrow \mathbb{R}$ be a spherically symmetric,
smoothed out characteristic function of the ball $B(0,\frac{R}{6})$
supported in the same ball. Let, as in \eqref{Phiaeq} - \eqref{phiAj},
\begin{equation}\label{psiAj}
\psi_{A_k}(x_{A_k}):=\psi_{k}(z_{A_k}),\ \quad \mbox{where}\ \quad  \psi_k(z_{A_k}):=\frac{(\phi_k \chi_R^{\otimes Z_k})(z_{A_k})}{\|\phi_k \chi_R^{\otimes Z_k}\|}.\end{equation}
 Here $ z_{A_k}:= 
(z_i :=x_i-y_k: i \in A_k)$, where, recall,  $x_{A_k}=(x_i: i \in A_k)$, and  $\chi_R^{\otimes |A_k|}(z_{A_k}):=\prod_{i\in A_k} \chi_R(z_i)$. Furthermore, we define
\begin{equation}\label{Psiaeq}
 \Psi_{a}(x_1,...,x_N):= \prod_{k=1}^M \psi_{A_k}(x_{A_k}).
\end{equation}
 From \eqref{Phiaeq}, \eqref{phiAj},
\eqref{Phiadecay}, $H_a \Phi_a=\Einfty \Phi_a$, where
$\Einfty$ was defined in \eqref{Einftydef}, and from the construction of $\Psi_a$ we obtain the estimates
\begin{align}\label{eigenf-differ}
&  \Phi_a  \doteq \Psi_a,\ \quad 
 H_a \Psi_a  \doteq \Einfty \Psi_a,\\
\label{Psiadecay}
&\|e^{\theta \langle x-y_a \rangle} \partial^{\alpha} \Psi_a\|
\lesssim 1, \forall \alpha, \text{ } 0 \leq |\alpha| \leq 2,\\
\label{pairor}
& \Psi_a  \Psi_b =0, \forall a,b \in \mathcal{A}^{at}, \text{ } a \neq b.
\end{align}
    We choose $P$ as  the orthogonal projection on  $\text{span}
\{\Psi_a: a \in \mathcal{A}^{at}\}$.  An important fact about $P$ is that it commutes with the permutations
\begin{equation}\label{Psym}
P T_\pi=T_\pi P, \quad \text{ } \forall \pi \in S_N,
\end{equation}
where, recall, $T_{\pi}$ are the unitary operators given by
\eqref{Tpi}. Moreover, according to \eqref{FP},
 we have to compute $P\H P$ and $U(\lambda)$. To this end we use
 \eqref{pairor} and as a consequence,  we have
\begin{equation}\label{P}
   P= \sum_{a \in \mathcal{A}^{at}} P_{a, R} = \sum_{a \in \mathcal{A}^{at}} P_{\Psi_a} 
\end{equation}
(acting on  the subspace $\mathcal{H}_A=\bigwedge_1^N L^2(\mathbb{R}^3) $). In the present context, $P_{a, R} =P_{\Psi_a}$ is the same the projection as in Introduction.   We use the notation $P_{a, R}$ rather than $P_{\Psi_a}$ in order to exposition similar to the one for arbitrary symmetry types in Section \ref{sec:proofvdWThm-s}. 

 To show that the Feshbach map exists, we note that $\Ran (P) \subset
\text{Dom}(H)$ since the range of $P$ is spanned by $(\Psi_a)_{a \in
\mathcal{A}^{at}}$ and by \eqref{Psiadecay} each $\Psi_a$ is in
$H^2(\R^{3N})$. Hence the condition (a) for the existence of the Feshbach map holds. We will prove in
Section \ref{Hbotbndseveral} that, under Property (E), there is a
$\imsgap > 0$, independent of $y$ s.t. for $R$ large enough, the following \textit{stability bound}  holds:
\begin{equation}\label{Hbotbnd}
\H^\bot \geq \Einfty+2\imsgap,
\end{equation}
where $\Einfty$ was defined in \eqref{Einftydef} and, recall,  the notation  $\H^\bot=P^\bot \H P^\bot$, acting on  the subspace $\mathcal{H}_A=\bigwedge_1^N L^2(\mathbb{R}^3) $. We
\textit{assume} this for now. Using \eqref{Hbotbnd}
 we obtain that $H^\bot-\lambda$ is invertible for all
$\lambda \leq E(\infty)+ \imsgap $.
 Thus the condition (b) is also
satisfied and, for all $\lambda \leq E(\infty)+\gamma$, the Feshbach
Schur map $F_P(\lambda)$ is well defined.

 Our goal now is to prove the following
 \begin{theorem}\label{thm:vdWlawcond} Assume 
the estimate \eqref{Hbotbnd} 
 holds. Then so does the van der
Waals law for the maximal spin (Theorem \ref{thm:vdW-maxspin}). 
\end{theorem}
To this end we have to estimate the two terms on the r.h.s of \eqref{FP}. We begin with the first one. 

\medskip

\paragraph{Estimate of $P \H P$.}

We show  that (on  the subspace $\mathcal{H}_A=\bigwedge_1^N L^2(\mathbb{R}^3) $)
\begin{equation}\label{PHP}
P \H P \doteq \Einfty P.
\end{equation}
Indeed,  using \eqref{P}, that $\Psi_a$ and $ \Psi_b$ have disjoint
supports and that $H$ is a local operator we obtain that
$ P \H P=\sum_{a \in \mathcal{A}^{at}} P_{a, R} \H P_{a, R}=\sum_{a \in \mathcal{A}^{at}}   \langle \Psi_a, \H \Psi_a \rangle P_{a, R}. $ 
Next using \eqref{P} and that by \eqref{eigenf-differ}, 
we have $H
\Psi_a= H_a \Psi_a+I_a \Psi_a \doteq \Einfty \Psi_a+ I_a \Psi_a$, $\forall a \in \mathcal{A}^{at}$, we
obtain that
\begin{equation}\label{PHP2}
P \H P \doteq \Einfty P + \sum_{a \in \mathcal{A}^{at}} 
\langle \Psi_a, I_a \Psi_a \rangle P_{a, R}.
\end{equation}
Note that since the group $S_N$ acts transitively on partitions from $\mathcal{A}^{at}$, the r.h.s. of this expression commutes with the permutations $T_\pi,\ \pi\in S_N$, and therefore leaves  the subspace $\mathcal{H}_A=\bigwedge_1^N L^2(\mathbb{R}^3) $ invariant. Eq. \eqref{PHP} follows from \eqref{PHP2} and the following 
\begin{lemma}\label{lem:newton}
For all $a \in \mathcal{A}^{at}$, we have
\begin{equation}\label{psiaIbpsib}
\langle \Psi_{a}, I_a \Psi_{a} \rangle =\Tr (I_a P_{a, R}) = 0.
\end{equation}
\end{lemma}
\begin{proof}
 Let $a=\{A_1,...,A_M\}$.
Due to charge neutrality for each atom $A_j$, we have that
\begin{equation}\label{Iabreakup}
I_{a}=e^2\sum_{i<j}^{1,M} 
\sum_{k \in A_i, l \in A_j} 
[-\frac{1}{|y_j-x_k|}-\frac{1}{|y_i-x_l|}+\frac{1}{|y_i-y_j|}+\frac{1}{|x_k-x_l|}].
\end{equation}
Let   $\rho_{A, R}$  be  the one-electron density  associated with the projection $P_{A, R}$ (by \eqref{psiAj} and \eqref{Psiaeq}, $\rho_{A, R}(z)=\chi(z)\rho_{A}(z)$, where $\rho_{A}(z)$ is defined in \eqref{rhoA-def}). Using \eqref{Psiaeq}, or $ P_{a, R}=\prod_i P_{\psi_i}$,  
\eqref{Iabreakup}, and passing to the variables
\begin{equation}\label{z-var} 
z_{km}=x_k-y_m,\ \forall  l \in A_m,\  \forall m=1,...,M,
\end{equation} 
we obtain \eqref{TrIP} - \eqref{tildeIijkl}.
By Proposition \ref{prop:spherical}  and the choice of  the cut-off, the one-electron densities $\rho_{A_j, R}(z)=\chi(z)\rho_{A}(z)$ are spherically symmetric. This, together with Newton's theorem (see for example \cite{LL}) and support properties due to the cut-off (e.g. $\int \rho_{A_i, R} (z)\frac{1}{|y_{ij}+z|} dz=\frac{1}{|y_{ij}|} \int_{|z| \leq |y_{ij}|} \rho_{A_i, R} (z) dz= \frac{1}{|y_{ij}|}$,  $y_{ij}=y_i-y_j$), implies \eqref{psiaIbpsib}.
\end{proof}

\paragraph{Remark.}  By the charge neutrality of the clusters $A_j$, without spherical symmetry, the leading term in $\Tr (I_a P_{a, R})$ is of the order at least  $R^{-3}$ if at least two atoms have non-zero dipoles and of a higher order otherwise (see \eqref{Iijkl-exp}).

\medskip
Eq. \eqref{PHP} and  the variational inequality $\E P \leq P \H P$ imply 

\begin{equation} \label{ineqE}
 \E\ \dot\le\ \Einfty .
\end{equation}
By \eqref{Hbotbnd} and \eqref{ineqE}, the Feshbach map is well defined for $E$ and by
\eqref{FP}  - \eqref{FSE} and \eqref{PHP}, we have that  $E$ is an eigenvalue of 
\begin{equation}\label{FPEdot}
F_P(E) \doteq \Einfty+U(E)|_{\Ran P}
\end{equation}
(acting on  the subspace $\mathcal{H}_A=\bigwedge_1^N L^2(\mathbb{R}^3) $).  
  Now we estimate $U(E)$.
\paragraph{Estimate of $U(\E)$}

 We recall that $H_{A}$ denotes the Hamiltonian of 
an atom $A$ and $E_A$ denotes its ground state energy. We let $P_{ij}$ stand for
the projection
onto the ground state of $H_{A_i}+H_{A_j}$ and $P_{ij}^\bot := \one -P_{ij}$. 
Let  $|y_{ij}|=|y_i-y_j|$.
 Our goal is to show
  \begin{lemma} 
  There are positive constants $\sigma_{ij}$, independent of $y$ and depending only on the nature of the atoms $i$ and $j$,   s.t.
  \begin{equation}\label{U-est}
U(\E) = f(y) P +O(\frac{1}{R^7}),  \quad 
\mbox{where}  \quad f(y)=\sum_{i<j}^{1,N} \frac{e^4
\sigma_{ij}}{|y_{ij}|^6}. 
\end{equation}
Moreover,  $\sigma_{ij}$ are given by the equations
\begin{equation}\label{sigmaijvalue}
\sigma_{ij}=\langle f_{ij} \phi_{i} \phi_{j}, R_{ij}^\bot f_{ij} \phi_{i}
\phi_{j} \rangle,
\end{equation}
where $R_{ij}^\bot:=(P_{ij}^\bot(H_{A_i}+H_{A_j}) P_{ij}^\bot- E_{A_i}-E_{A_j})^{-1}$  and 
 \begin{equation} \label{fij}
 f_{ij}:=\sum_{k \in A_i,l \in A_j}  
 [z_{ki} \cdot z_{lj}- 3(z_{ki} \cdot \widehat{y_{ij}})(z_{lj} \cdot  \widehat{y_{ij}})].\end{equation}
\end{lemma}
\begin{proof} By \eqref{Hbotbnd} and  \eqref{ineqE}, the operator $H^{\bot}-E$,   where, recall,  $H^{\bot}=P^{\bot} H  P^{\bot}$, has a bounded inverse, which we denote by 
$R^{ \bot}(E):=(H^{\bot}-E )^{-1}$.  It follows from \eqref{P}, \eqref{psiaIbpsib} and \eqref{pairor} and from  \eqref{P}, \eqref{pairor} and \eqref{eigenf-differ} that
\begin{equation}\label{chiaest2}
P^\bot I_a \Psi_a= I_a \Psi_a\ \quad \mbox{and}\ \quad 
 P^\bot H_a \Psi_a \doteq 0.
\end{equation}
Using \eqref{P} and \eqref{chiaest2} we obtain that
\begin{equation}\label{PorthHP}
 P^\bot \H P \doteq \sum_{a \in \mathcal{A}^{at}}  I_a P_{a, R}.
\end{equation}
 The relations \eqref{U}, \eqref{Hbotbnd}, \eqref{PorthHP} and $P H P^\bot= (P^\bot H P)^*$ give
 that
  \begin{align}\label{U-aprox1}
U (E) \dot = \sum_{a,b} 
 P_{a, R}  I_a   P^\bot R^{\bot}(E) P^\bot  I_b P_{b, R}. 
\end{align}
By the same remark, as the one after \eqref{PHP2}, the r.h.s. of this expression commutes with the permutations $T_\pi,\ \pi\in S_N$, and therefore leaves  the subspace $\mathcal{H}_A=\bigwedge_1^N L^2(\mathbb{R}^3) $ invariant. 

Next, we estimate $I_a \Psi_a$, which is done in the next lemma.
\begin{lemma}
\begin{equation}\label{IaPsia-est}
\|I_a P_{a, R}\| \lesssim \sum_{i<j}^{1, M} \frac{1}{|y_i-y_j|^3}.
\end{equation}
\end{lemma}
\begin{proof}
Let $a=\{A_1,...,A_M\}$. Recall that, due to charge neutrality for each atom $A_j$, be have that by  \eqref{Iabreakup}  and \eqref{tildeIijkl}, the intercluster interaction $I_{a}$ can be written as
 $I_{a}= \sum_{i<j} 
   \sum_{k \in A_i, l \in A_j}I_{ij}^{kl}$, 
  with $I_{ij}^{kl}=I ( z_{ki}, z_{lj}; y_{ij})$ in the variables \eqref{z-var}, where $I$ is given in \eqref{tildeIijkl}. 
 Using that
  \begin{equation}\label{Taylorsevmu}
  |z \pm y_{ij}|^{-1}=|y_{ij}|^{-1}(1 \pm \frac{2 z \cdot
  \widehat{y_{ij}}}{|y_{ij}|}+\frac{|z|^2}{|y_{ij}|^2})^{-\frac{1}{2}},
  \end{equation}
 where, recall,  $y_{ij}=y_i-y_j$, we expand the function $ I_{ij}^{kl}$, given in \eqref{tildeIijkl},  in $z_{ki}$ and $z_{lj}$, we obtain
  \begin{equation*}
 I_{ij}^{kl}= \frac{e^2}{|y_{ij}|}\left[-\frac{(z_{ki}-z_{lj}) \cdot
\widehat{y_{ij}}}{|y_{ij}|}+\frac{z_{ki} \cdot
\widehat{y_{ij}}}{|y_{ij}|}-\frac{z_{lj} \cdot
\widehat{y_{ij}}}{|y_{ij}|}-\frac{|z_{ki}-z_{lj}|^2-3((z_{ki}-z_{lj})
\cdot \widehat{y_{ij}})^2}{2
|y_{ij}|^2}\right]$$$$+\frac{e^2}{|y_{ij}|}
\left[\frac{|z_{ki}|^2-3(z_{ki} \cdot \widehat{y_{ij}})^2}{ 2|y_{ij}|^2}+\frac{ |z_{lj}|^2-3(z_{lj} \cdot \widehat{y_{ij}})^2}{2|y_{ij}|^2}\right]+O(\frac{|z|^3}{|y_{ij}|^4}),
  \end{equation*}
where $|z|:=\max_{k, l} (|z_{ki}|, |z_{lj}|)$, provided $|z|\le \frac{1}{3} |y_{ij}|$, together with the fact that 
$\supp \Psi_a \subset \{|z|\le \frac{1}{3} |y_{ij}|\}$ gives that
 \begin{equation}\label{Iijkl-exp} 
\sum_{k \in A_i, l \in A_j} I_{ij}^{kl}=\frac{e^2}{|y_{ij}|^3}f_{ij}(z,\widehat{y_{ij}})+O(\frac{|z|^3}{|y_{ij}|^4})\
 \text{ on } \supp \Psi_a, \end{equation}
where $f_{ij}(z,\widehat{y_{ij}})$ is given in \eqref{fij}. This estimate together with \eqref{Psiadecay} gives \eqref{IaPsia-est}.
\end{proof}

For each $a \in \mathcal{A}^{\at}$, we introduce the sets 
\begin{equation}\label{hatOma}
\hat \Om^\nu_a :=   
 \{x \in \mathbb{R}^{3N}: |x_i-y_m| \le  \nu R,\ 
\forall i \in A_m,\  \forall 
A_m\in a \}, 
 \end{equation}
and  smoothed out characteristic functions,  
$\chi_a,\ a \in \mathcal{A}^{\rm at},$ symmetric under all permutations in $S(a)$ and such that
\begin{equation}\label{chia-supp}
\chi_a(x)=  1\   \text{ on }\ \hat \Om^{1/6}_a  
 \text{ and  }\ =0   \text{   outside 
  }\  \hat \Om^{1/5}_a,  
\end{equation}
\begin{equation}\label{chiader'}
|\partial^\alpha \chi_a | \lesssim R^{-|\alpha|}, \text{ for any
multi-index } \alpha.
\end{equation}
Eq. \eqref{chia-supp}  and the definition of $P_{a, R}$ imply 
\begin{equation}\label{chiaIa} I_a (-\Delta +1)^{-1/2}=O(\frac1R) \text{ on } \supp \chi_a.\end{equation}
\begin{equation}\label{chiaPb}
\chi_a P_{b, R}  = \delta_{ab}  P_{a, R}, \quad \forall a,b \in \mathcal{A}^{at}.
 \end{equation} 

Let  $ R_a^{ \bot}(E):=(H_a^{\bot}   -E)^{-1} $, where  $H_a^\bot :=H_a P_{a}^\bot$, with the notation 
$P_a=P_{\Phi_a}:= |\Phi_a \rangle \langle \Phi_a|$ (the orthogonal projection onto the ground state of $H_a$).  Our goal is to show the following 
estimate 
\begin{equation}\label{charcom}
\chi_a R^{\bot}(E)- R_a^{ \bot}(E)\chi_a =O(\frac{1}{R}).
\end{equation}
Factoring out the inverse operators on  the l.h.s. of \eqref{charcom}, 
  we obtain that
 \begin{equation}\label{B}
\chi_a R^{\bot}(E)- R_a^{ \bot}(E)\chi_a= R_{a}^{\bot}(E) 
V R^{ \bot}(E), 
 \end{equation}
 where $V:=   P_a^{ \bot} H_a \chi_a- \chi_a P^{\bot} H P^{\bot}$. 
Due to the cut off of the ground states 
 we have that
 Eq. \eqref{chiaPb} and the definition of  $P^{\bot}$ give  
 \begin{equation}\label{chiaP}\chi_a P^{\bot}  =  P_{a, R}^{\bot}\chi_a\  \text{   and        }\  [\chi_a, 
 P_{a, R}^{ \bot}]=0.  \end{equation}   
The latter relations, together with 
the fact that $H_a$ is a local operator, imply that 
$ \chi_a  P^{\bot} H P^{\bot}=  P_{a, R}^{\al \bot} \chi_a H,$
 which, together with $ H= H_a+  I_a$, 
 gives   
 \begin{equation}\label{charsimp'} V =     P_{a, R}^{ \bot} ( H_a \chi_a- \chi_a H_a+ \chi_a I_a)+(P_{a, R}-P_a)H_a \chi_a.  \end{equation}
   Using \eqref{chiader'}, \eqref{chiaIa}, $P_{a, R}\dot = P_{a}$ and the $H_a$ boundedness of thegradient, $\n$, we obtain that
   $\| V (-\Delta +1)^{-1/2}\| \lesssim \frac{1}{R},$
 which together with \eqref{B} implies \eqref{charcom}. 

 Since $\chi_a$ commutes with $I_a$, we have $ P_{a, R} I_a    = P_{a, R}   I_a   \chi_a $. Using this to insert $\chi_a$ into \eqref{U-aprox1},  and using \eqref{chiaP}, \eqref{charcom}, \eqref{IaPsia-est}  and \eqref{chiaPb} 
 gives 
  \begin{align}\label{U-aprox2} U (E) &= \sum_{a,b } 
 P_{a, R}  I_a     R_a^{ \bot}(E)\chi_a  I_b P_{b, R}  +O(\frac{1}{R^7}))\notag \\ 
  &= \sum_{a} 
P_{a, R} I_a   R_a^{\bot}(E) I_a P_{a, R}    
+O(\frac{1}{R^7}).
\end{align}
Again, it is easy to see that the r.h.s. of this expression commutes with the permutations $T_\pi,\ \pi\in S_N$, and therefore leaves  the subspace $\mathcal{H}_A=\bigwedge_1^N L^2(\mathbb{R}^3) $ invariant. 

Now, by a standard estimate and \eqref{IaPsia-est}, we have $0\le \lan \psi, \sum_{a}P_{a, R} I_a   R_a^{\bot}(E) I_a P_{a, R} \psi\ran\ls \frac{1}{R^6}\lan \psi, \sum_{a}P_{a, R} \psi\ran$, which by \eqref{P} and   \eqref{U-aprox2}, gives $\lan \psi, U (E) \psi\ran\ls \frac{1}{R^6}\|P\psi\|^2$, which in turn implies $\| U (E) \big |_{\Ran P} \|\ls \frac{1}{R^6}$. Using, this estimate, together with  \eqref{FPEdot} 
 and the fact that $E$ is an eigenvalue of $F_P(E)$, we estimate
\begin{equation}\label{EEinfty-bnd}
|E- \Einfty|\ls R^{-6}.
\end{equation}
Using this and \eqref{IaPsia-est} in  \eqref{U-aprox2}, and using the notation $R_a^{\bot}=R_a^{\bot}(E(\infty))$, we obtain
\begin{align}\label{Uab} 
U (E) &= \sum_{a}U_{aa}  +O(\frac{1}{R^7}),\ U_{aa} =
P_{a, R} I_a   R_a^{\bot} I_a P_{a, R}    
\end{align}


\DETAILS{Now we show that $U_{aa}$ is independent of $a$. First, recall that for any $a,b \in
\mathcal{A}^{at}$ there exists a permutation $\pi$ with $b=\pi a$.

Moreover, $T_{\pi}\Psi_a=\Psi_{\pi a}$. The last two facts imply
that
$T_{\pi} I_a \Psi_a=I_b \Psi_b. $ 
 Using \eqref{Uab}, and the last relation we obtain that
$$U_{bb}(E)=\langle T_\pi I_a \Psi_a, (\H^\bot-E)^{-1} T_\pi I_a \Psi_a \rangle.$$
Since $H$ and $P$ commute with $T_\pi$ (see \eqref{Hsym} and
\eqref{Psym}) and since $T_\pi$ is unitary, the last relation
implies that \begin{align}\label{Uab2} U_{bb}(\E)=U_{aa}(\E). \end{align}
Thus it suffices to estimate $U^{aa}(\E)$ for a fixed $a$. }

We insert  
\eqref{Iabreakup} with \eqref{Iijkl-exp} 
into the expression for $ U_{aa}$ in \eqref{Uab} 
and use that $O(\frac{|z|^3}{|y_{ij}|^4})P_{a, R}=O(\frac{1}{|y_{ij}|^4})$, to obtain 
\begin{equation}\label{Uab3}
U_{aa}=\sum_{ij,\beta} \frac{e^4 \sigma_{ij,kl}^{a}}{|y_{ij}|^3 |y_{kl}|^3} P_{a, R}+O(\sum_{ij,kl}\frac{e^4}{|y_{ij}|^3 |y_{kl}|^4}),
\end{equation}
where $ij, kl$ run in pairs of nuclei in the decomposition $a$ and $\sigma_{ij, kl}^{a}=\langle  f_{ij} \Phi_a, R_a^\bot  f_{kl} \Phi_a \rangle$.  

 Finally, we prove the following lemma
\begin{lemma}
\begin{equation}\label{sigest}
 \sigma_{ij,kl}^a =\sigma_{ij}\del_{ij, kl} +O(\frac{1}{R}),
\end{equation}
where $\sigma_{ij}$ 
are given in \eqref{sigmaijvalue}.
\end{lemma}
\begin{proof}
 Consider an atomic decomposition $a=(A_1,...,A_M)$. 
 Recall the notation from Section \ref{sec:prelim} and write the orthogonal projection $P_a$ on $\Phi_a$ as $P_{a}=P_{kl}\otimes P^{kl}$, where 
$P_{ij}=P_{A_i}\otimes P_{A_j}$ and $P^{ij}=\prod_{k\ne i, j} P_{A_k}$, with $P_{A_i}=P_{\phi_{A_i}}$, the orthogonal projection on $\phi_{A_i}$. Then  we have $P_{a}^\bot P^{kl}=P_{kl}^\bot\otimes P^{kl}$ and therefore, on the invariant subspace $\Ran P^{kl}$, 
\begin{align}\label{Habot}
 H_a P_a^\bot =P_{kl}^\bot (H_{A_k}+H_{A_l}) P_{kl}^\bot +\sum_{i\ne k, l} E_{A_i}.
\end{align}
This, the notation $R_{kl}^\bot:=(P_{kl}^\bot (H_{A_k}+H_{A_l}) P_{kl}^\bot-E_{A_k}-E_{A_l})^{-1}$ and the fact that $f_{kl}\Phi_a= P^{kl}f_{kl}\Phi_a \in \Ran P^{kl} $ and $\Einfty=\sum_{j=1}^N E_{A_j}$ give
 \begin{align}\label{Habot-Einv}
R_a^\bot (E) f_{kl}\Phi_a  &= R_{kl}^\bot f_{kl} \Phi_a,
\end{align}
which implies 
\begin{equation}\label{eq-W}
 \langle f_{ij} \Phi_a, R_a^\bot  f_{kl} \Phi_a \rangle
 =\langle f_{ij} \Phi_a, R_{kl}^\bot f_{kl} \Phi_a \rangle. 
\end{equation}

We now prove that
\begin{equation}\label{W-expr} \langle f_{ij} \Phi_a, R_{kl}^\bot f_{kl} \Phi_a \rangle= \left\{
\begin{array}{ccc} 0, \text{ if } ij \neq kl
\\ \sigma_{ij}, 
\text{ if }\ ij=kl
\end{array} \right.,\end{equation}
where $\sigma_{ij}$ are given in \eqref{sigmaijvalue}. 
We denote the l.h.s. of \eqref{W-expr} by $W_{ij;kl}$. If $ij \neq kl$, then we may assume without loss of generality that
$i \neq k,l$. Observing that the inner product on the right hand
side of \eqref{eq-W} reduces to an integral that is odd
in $z_{mi}, m \in A_i$ we obtain $W_{ij;kl}= 0$ for $ij \neq kl$. For the
case $ij =kl$, $W_{ij;kl}= \sigma_{ij}$ follows immediately from the
relations \eqref{Phiaeq} - \eqref{phiAj} and the fact that the
ground states of the atoms are normalized. Hence  $W_{ij;kl}$ follows.  

This, together with 
the definition $\sigma_{ij, kl}^{a}=\langle  f_{ij} \Phi_a, R_a^\bot  f_{kl} \Phi_a \rangle$ and 
\eqref{eq-W}, proves  \eqref{sigest}.
\end{proof}

The relations \eqref{Uab}, \eqref{Uab3} and \eqref{sigest}, together with \eqref{P},  imply \eqref{U-est}
modulo the fact that $\sigma_{ij}$ are positive and independent of
$y$ which we prove below.

First we note that the only dependence of $\sigma_{ij}$ on
$\widehat{y_{ij}}:=\frac{y_i-y_j}{|y_i-y_j|}$ appears on $f_{ij}$. To display this dependence 
we will write $f_{ij}^{\widehat{y_{ij}}}$ and $\sigma_{ij}^{\widehat{y_{ij}}}$ for $f_{ij}$ and $\sigma_{ij}$. For any rotation $R$ in
$\mathbb{R}^3$ we define $T_R$ acting on the space
$L^2(\mathbb{R}^{3(|A_i|+|A_j|)})$ (recall that $|A_k|$ is the
number of electrons of the $k$-th atom) of functions of the
variables $z_{kl}=x_k-y_l,\ k=1, \dots , |A_i|+|A_j| ,\ l =1, 2$, by
rotating them.   Since $T_R$ commutes with $H_{A_i}+H_{A_j}$ and since $\phi_i \phi_j$ is a ground state of $H_{A_i}+H_{A_j}$, we have that
 $T_R \phi_i \phi_j$ is also a ground state of $H_{A_i}+H_{A_j}$.
 Therefore, since the ground state is non degenerate and $T_R$ is unitary we obtain that
 $T_R \phi_i \phi_j=c(R) \phi_i \phi_j$ where $|c(R)|=1$. This
 implies 
$$T_R P_{ij}^\bot T_R^{-1}=P_{ij}^\bot.$$
Using the last relation and the fact that $H_{A_i}+H_{A_j}$ commutes
with $T_R$,
 together with \eqref{sigmaijvalue} and the fact
that $T_R$ is unitary, we obtain that
\begin{equation}\label{sigmaijrotation}
\sigma_{ij}^{\widehat{y_{ij}}}= \langle T_{R^{-1}}
(f_{ij}^{\widehat{y_{ij}}} \phi_i \phi_j), R_{kl}^\bot T_{R^{-1}}
(f_{ij}^{\widehat{y_{ij}}} \phi_i \phi_j) \rangle.
\end{equation}
 On the other hand using \eqref{fij} and the fact that $R$ is
 unitary we obtain that
\begin{equation}\label{fijrotation}
T_R^{-1}f_{ij}^{\widehat{y_{ij}}}=f_{ij}^{\widehat{R^{-1} y_{ij}}}.
\end{equation}
But using  that $T_R^{-1} \phi_i \phi_j = c(R^{-1}) \phi_i \phi_j$,
with $|c(R^{-1})|=1$, together with  \eqref{fijrotation}, we obtain that
$$T_{R^{-1}}  (f_{ij}^{ \widehat{y_{ij}}} \phi_i \phi_j)=c(R^{-1}) f_{ij}^{\widehat{R^{-1} y_{ij}}} \phi_i \phi_j.$$
The last relation together with \eqref{sigmaijrotation} and $|c(R^{-1})|=1$ implies that
$$\sigma_{ij}^{\widehat{y_{ij}}}=\langle f_{ij}^{\widehat{R^{-1} y_{ij}}} \phi_i\phi_j, R_{kl}^\bot (f_{ij}^{\widehat{R^{-1} y_{ij}}} \phi_i \phi_j)
\rangle=\sigma_{ij}^{\widehat{R y_{ij}}},$$ which implies that
$\sigma_{ij}$ is independent of $y_{ij}$.

Since the operator $P_{ij}^\bot (H_{A_i}+H_{A_j}) P_{ij}^\bot- E_i-E_j$ is a positive operator, then so is its inverse, $R_{kl}^\bot \geq 0$. Hence, by \eqref{sigmaijvalue}, we have that $\sigma_{ij} > 0$. 
\end{proof}

\paragraph{Conclusion of the argument.}

 From 
\eqref{FPEdot}  and \eqref{U-est} we have that
 \begin{equation}\label{FPalmostiden}
F_P(\E) =  \Einfty- f(y)+O(\frac{1}{R^7}). 
 \end{equation}
  Since the ground state energy $\E$ is an eigenvalue of $\H$, we
  conclude from   \eqref{FSE}, 
  \eqref{FPalmostiden}  and the definition $W(y)=\E-\Einfty$ that \eqref{vdWlaw}  holds.
This proves 
Theorem \ref{thm:vdWlawcond}. $\Box$

\paragraph{Proof of the necessity of Property (E).}\label{enecessity}
In this paragraph we will show that if Property (E) fails to hold then
so does the van der Waals - London law.
\DETAILS{When Property (E) fails, we have to define  the interaction
energy as
\begin{equation*}
W(y)=\E-\E_{\min} ,\ \mbox{where}\ \E_{\min}:=\min_{a \in
\mathcal{A}}\inf \sigma(H_a).
\end{equation*}}
%
We denote by $\mathcal{A}^{\min}$ the set of all $a \in \mathcal{A}$ for which 
$\inf \sigma(H_a)=E(\infty)$, where $E(\infty):=\min_{b\in \cA} \inf \sigma(H_b)$, and $\inf \sigma(H_a)$ is an isolated eigenvalue. 
By the HVZ and Zhislin theorems, this set is non-empty. If Property (E) holds if and only if
$\mathcal{A}^{\min}=\mathcal{A}^{at}$. Now assume that Property (E)
fails. 

To prove that the van der Waals London law also fails we use, as
before, the Feshbach map but with $P$,  the orthogonal projection on
$\text{span}\{\Psi_a: a \in \mathcal{A}^{\min}\}.$ 
\DETAILS{Note that the
condition $\inf(H_a)=E_{\min}$ implies that $E_{\min}$ is an
isolated eigenvalue of $H_a$. Indeed, we know that all positive ions
and neutral atoms have isolated ground states. Assume that some $-k$
ion, $k>0$,  does  not have isolated ground state energy. We may
assume without loss of generality that this ion is at $y_1$.
 Then $E_{1,-k}=E_{1,-k+1}$ i.e. we can remove an electron with no
energy cost. On the other hand we may assume that at $y_2$ there is
a negative ion. Then $E_{2,m}>E_{2,m-1}$ i.e. we would gain some
energy by putting an electron into the ion. Thus $E_{1,-k}+E_{2,m}>
E_{1,-k+1}+ E_{2,m-1}$ contradicting the condition
$\inf(H_a)=E_{\min}$. Therefore, every ion has a ground state energy
which, moreover, is exponentially decaying.}
 Proceeding as in the proof of \eqref{P} we can show that
$P = \sum_{a \in \mathcal{A}^{\min}} P_{a, R}.$ 
Furthermore, repeating the arguments of the proof of \eqref{Hbotbnd}
(with $\mathcal{A}^{at}$ replacing $\mathcal{A}^{\min}$) we can
prove that there exists $\gamma>0$ such that 
$H^\bot \geq \E_{\rm min}+2 \gamma, $ 
where $E_{\min}=\min_{a \in \mathcal{A}} \inf \sigma(H_a)$, which implies that the Feshbach Schur method can be used.

 We will now estimate $PHP$. Using the relation $P = \sum_{a \in \mathcal{A}^{\min}} P_{a, R}$ and proceeding as
 in the proof of \eqref{PHP2}, we can show that
 \begin{equation*}
P \H P \doteq \E_{\min} P+\sum_{a \in \mathcal{A}^{\min}} 
\langle \Psi_a, I_a \Psi_a \rangle P_{a, R}.
\end{equation*}
 Since the Property (E) is not satisfied it means that there exists
$a \in \mathcal{A}^{\min}$ with $a \notin \mathcal{A}^{at}$. For
such a decomposition $a=(I_1, \dots, I_M)$, we define the charges
$q_{ai}:=(Z_i-|I_i|)e$. Next, we apply Newton's theorem to show that
\begin{equation*}
\langle \Psi_a, I_a \Psi_a \rangle = \sum_{i\ne
j}\frac{q_{ai}q_{aj}}{|y_{i}-y_{j}|},\ 
,\ \forall a \notin \mathcal{A}^{at}.
\end{equation*}
Clearly, $\sum_{i\ne j}\frac{q_{ai}q_{aj}}{|y_{ai}-y_{aj}|}\asymp \frac{1}{R}$. 
The last three relations imply that
\begin{equation}\label{PHPo}
P \H P \doteq \E_{\min} P+\sum_{a \in
\mathcal{A}^{\min}/\mathcal{A}^{at}} \sum_{i\ne
j}\frac{q_{ai}q_{aj}}{|y_{i}-y_{j}|} P_{a, R}. 
\end{equation}

To estimate $U(\lam)$, we proceeding as in the proof of \eqref{U-aprox1}, 
to obtain that
\DETAILS{\begin{equation*}
U(\lambda) \doteq \sum_{a,b \in \mathcal{A}^{\min}}P_{a, R} I_a  P^\bot (H^\bot-\lambda)^{-1} P^\bot I_b P_{b, R}.
\end{equation*}
In addition proceeding as in 
proving \eqref{IaPsia-est}, 
we can show that
\begin{equation}
\|I_a P_{a, R}\| \lesssim \frac{1}{R}, \forall a \in \mathcal{A},
\end{equation}
The last two relations together with 
$H^\bot \geq \E_{min}+2 \gamma,$ for $\gamma>0$, give}
\begin{equation}\label{Ulambdao}
U(\lambda)=O(\frac{1}{R^2}),\ \quad \forall \lambda \leq
\E_{\min}+\gamma.
\end{equation}
 The relations \eqref{FP}, \eqref{PHPo} and
\eqref{Ulambdao} give give, $\forall \lambda \leq \E_{\min}+\gamma,$ that
\begin{equation*}
F_P(\lambda)=E_{\min}P+\sum_{a \in \mathcal{A}^{\min}/\cA^{at}}\sum_{i\ne j}
\frac{q_{ai}q_{aj}}{|y_{i}-y_{j}|} P_{a, R}+O(\frac{1}{R^2}), 
\end{equation*}
 and therefore, by $\sum_{i\ne j}\frac{q_{ai}q_{aj}}{|y_{i}-y_{j}|}\asymp \frac{1}{R}$ and \eqref{FSE}, the van der Waals law fails. As a consequence
 Property (E) is necessary for van der Waals law to hold.



 \section{Proof of stability bound \eqref{Hbotbnd} assuming Property (E)} \label{Hbotbndseveral}
Let $ E_{a}^{'}$ be the first excited state energy of the hamiltonian $H_a$
and $\Einfty$ as defined in \eqref{Einftydef}.  We define
\begin{equation}\label{L1def}
\gamma_1=\min_{a\in \cA^{\rm at}} E_{a}^{'}-  \Einfty, 
\quad
\gamma_2=\min_{a\in \cA/\cA^{\rm at}} E_a - \Einfty 
\end{equation}
(recall that $\Einfty=E_a$ for $a\in \cA^{\rm at}$). By Property (E) we have that $\gamma_2>0$.
  Finally, we let
 \begin{equation}\label{Ldef}
 \gamma_0:=\min\{\gamma_1,\gamma_2\} >0.
 \end{equation}
 The stability bound \eqref{Hbotbnd} follows immediately from the
 following proposition:
\begin{proposition}\label{Hbotbndproposition}
There exists a constant $C>0$ such that $\H^\bot
\geq E(\infty)+\gamma_{0}- \frac{C}{R}$. In particular,
\eqref{Hbotbnd} holds for $R$ large enough.
\end{proposition}
\begin{proof}
 The idea here is to localize the Hamiltonian $H$ to different domains of the configuration space $\R^{3N}$, each of which corresponding to a different break-up of the system into independent subsystems,  and use geometry of these domains to estimate $H$. Since $H$ is a differential operator, we have to pay a prize for this. A neat way way of doing such a localization is by constructing - following \cite{Sig} - a partition of unity $\{J_a\}$, labeled by  decompositions from $\mathcal{A}$ (see Subsection \ref{sec:deco}) and satisfying $\sum_{a \in \mathcal{A}}J_{a}^2=1$ and then using the IMS localization formula 
 (see for example \cite{CFKS})
\begin{equation}\label{IMS}
\H=\sum_{a \in \mathcal{A}} \left( J_{{a}} \H J_{{a}}-|\nabla J_{a}|^2 \right).
\end{equation}

 We will split the configuration space $\mathbb{R}^{3 N}$ into domains. For each
$a=(A_1,...,A_M) \in \mathcal{A}$, we define
\begin{equation}\label{Omegadef}
\Omega_{a}^\nu=\{(x_1,...,x_N):|x_i-y_j| \geq \nu R,\ \forall
j=1,...,M ,\ \forall i \notin A_j \}.
\end{equation}
Following \cite{Sig} we will now construct a partition of unity
$(J_a)_{a \in \mathcal{A}}$ having the properties:
\begin{equation}\label{partunprop1}
0 \leq J_{a} \leq 1, \supp J_{a} \subset \Omega_{a}^{\frac{1}{5}},\
\|\p^\al J_{a}\|_{L^{\infty}} \lesssim R^{-|\al|}, 
\end{equation}
\begin{equation}\label{symmetryJa} 
J_a\ \mbox{commute with all elements of}\ S(a),
\end{equation}
where, recall, $S(a)$ is the group of permutations that keeps the
clusters of $a$ invariant,  for all $a \in \mathcal{A}$, and
\begin{equation}\label{partunprop2}
\sum_{a \in \mathcal{A}}J_{a}^2=1.
\end{equation}
 We consider the functions
$F_{a}=\chi_{\Omega_{a}^{\frac{3}{10}}}*\phi$ where $\phi: \R^{3N}
\rightarrow \R$ is a $C_c^{\infty}$ function supported in
$B_{\frac{1}{10}}(0)$ (ball of radius $\frac{1}{10}$ centered at the
origin) symmetric with respect to particle coordinates, with $\phi
\geq 0$ and $\int_{\R^{3N}} \phi=1$ and $\chi_{A}$ denotes the
characteristic function of the set $A$. Then $F_{a} \in C^{\infty}$
and $F_{a} \geq 0$. Furthermore using the triangle inequality and
the fact that $\phi$ is supported in $B_{\frac{1}{10}}(\vec{0})$ we
obtain that $\supp(F_{a}) \subset \Omega_{a}^{\frac{1}{5}}$ and that
$F_{a}|_{\Omega_{a}^{\frac{2}{5}}}=1$. The last relation together
with the fact that $\cup_{a \in \mathcal{A}}
\Omega_{a}^{\frac{2}{5}}=\R^{3N}$ gives that $\sum_{a \in
\mathcal{A}} F_{a} \geq 1$. Therefore, if we define
\begin{equation*}
J_{a}=\frac{F_{a}}{\sqrt{\sum_{a \in \mathcal{A}} F_{a}^2}},
\end{equation*}
then $J_{a}$ is a partition of unity satisfying all the desired
properties.


Now we use the IMS localization formula 
\eqref{IMS}, together with \eqref{partunprop1}, to obtain
$\H \geq \sum_{a \in \mathcal{A}} J_{{a}} \H J_{{a}}-\frac{C}{R}.$ 
 The last relation together with the secomposition $H =H_a+I_a$ (see \eqref{Hadecomp})
  and $I_a \ge - \sum_{j=1}^M \sum_{i=1, i \notin A_j}^N
\frac{e^2 Z_j}{|x_i-y_j|}\ge -\frac{C}{R}$ on $ \supp J_{a}$ implies
that
\begin{equation}\label{Hbotest'}
\H \geq \sum_{a \in \mathcal{A}} J_{{a}} \H_a J_{{a}}-\frac{C}{R}. \end{equation} 
The last relation together with the fact that $P^\bot$ is self-adjoint and $P^\bot \leq 1$ implies that
\begin{equation}\label{Hbotgenest1}
P^\bot \H P^\bot \geq \sum_{a \in \mathcal{A}} P^\bot J_{{a}}H_{{a}}
J_{{a}} P^\bot-\frac{C}{R}.
\end{equation}
Proposition \ref{Hbotbndproposition} now follows from the equations  \eqref{partunprop2}, \eqref{Hbotgenest1} and
the lemma below which provides estimates - \eqref{PJHJPest} - of each of the terms on the right hand side of \eqref{Hbotgenest1}.
 \end{proof}
\begin{lemma}\label{lem:PJHJPest}
With $\gamma_{0}$ defined in \eqref{Ldef}, we have, $\forall a \in \mathcal{A},$ 
\begin{equation}\label{PJHJPest}  P^\bot J_{{a}}H_{{a}} J_{{a}} P^\bot \geq
(\Einfty+\gamma_{0})P^\bot J_{{a}}^2 P^\bot-\frac{C}{R}. 
\end{equation}
(Note that the restriction to  
$\mathcal{H}_A=\bigwedge_1^N L^2(\mathbb{R}^3) $, as  any restriction, only lifts the lower bound.)
\end{lemma}
\begin{proof}
We will consider the following two cases
\newline
\newline
\textbf{Case (1)}: $a \in \mathcal{A}/\mathcal{A}^{at}$ (i.e. some
clusters of the decomposition $a$ include ions). 
In this case,
\DETAILS{ we define $n_j:=Z_j-|A_j|$. Since by definition
$E_{j,n_j}$ is the infimum of the spectrum of the hamiltonian
$H_{A_j}$,  defined in \eqref{Ha1}, we have that $H_{A_j} \geq
E_{j,n_j}.$ The last inequality together with \eqref{L2def},
\eqref{Ldef} and \eqref{Ha} gives that
\begin{equation*}
H_{a} \geq E_{1,n_1}+...+E_{M,n_M}\geq(\Einfty+\gamma_{0})
\end{equation*}}
Condition (E) implies that $H_a\ge E(\infty)+\g_0$ and, as a consequence,
\begin{equation}\label{Hbotgencase1}
P^\bot J_{a} H_{a} J_{a} P^\bot \geq (\Einfty+\gamma_{1}) P^\bot
J_{a}^2 P^\bot.
\end{equation}

\noindent\textbf{Case (2)}: $a \in \mathcal{A}^{at}$ (i.e. $a$ is a decomposition of the system to $M$ (neutral) atoms).
 We  recall the notation  $P_{a, R}=P_{\Psi_a}$ and   $P_a=P_{\Phi_a}$. 
Using the  relation $H_a P_a^\bot  \ge (\Einfty+ \gamma_0) P_a^\bot$, we find
\begin{align}\label{IMS5'}
 J_a H_a J_a  &\geq  J_a [(\Einfty+ \gamma_2) P_a^\bot + \Einfty P_a]  J_a  \notag \\ &=  (\Einfty+ \gamma_2)  J_a^2 
  - \gamma_2  J_a P_a  J_a .
\end{align}
  By \eqref{eigenf-differ}, \eqref{Omegadef} and \eqref{partunprop1} and the support properties of $\Psi_b$, we have that
\begin{equation}\label{Janeqb}
\| P_{a}-  P_{a, R}\| \doteq 0\ \quad \mbox{and}\ \quad   J_a P_{b, R} =P_{b, R}   J_a = \del_{a, b}   P_{a, R}. 
\end{equation}
Therefore, by  \eqref{P},   we have $P^\bot  J_a P_a  J_a P^\bot \doteq P^\bot  P_{a, R} P^\bot  \doteq 0$. 
\DETAILS{
$P^\bot J_a H_a J_a P^\bot= P_{\Psi_a}^\bot J_a H_a J_a P_{\Psi_a}^\bot. $ 
By \eqref{Janeqb}, \eqref{partunprop1} and \eqref{partunprop2} we have that 
$J_a=1$ on the support of $\Psi_a$ and therefore $J_a, P_a^\bot$ commute. This, together with \eqref{IMS3}, implies that
\begin{equation}\label{IMS4}
P^\bot J_a H_a J_a P^\bot= J_a P_{\Psi_a}^\bot H_a
 P_{\Psi_a}^\bot J_a.
\end{equation}
The last equation together with}
%
Hence 
 \begin{equation}\label{IMS5}
P^\bot J_a H_a J_a P^\bot \dot\geq (\Einfty+ \gamma_2)  J_a^2,
\end{equation}
 which, together with 
 the previous case, implies \eqref{PJHJPest} for all $a \in \mathcal{A}^{at}$.
 \end{proof}

\section{Proof of Theorem \ref{thm:vdW-sig}}\label{sec:proofvdWThm-s}

\paragraph{The general set-up in the case of statistics.}\label{sec:setupstat} 
\DETAILS{We begin with the general setup for the proof of Theorem \ref{thm:vdW-sig}. 
   Let $S_N$ denote the permutation group of $\{1,2,...,N\}$. We
consider the representation $T: S_N \rightarrow
U(L^2(\mathbb{R}^{3N}))$, where $U(L^2(\mathbb{R}^{3N})),$ denotes
the unitary operators on $L^2(\mathbb{R}^{3N})$ defined by the
equation
\begin{equation}\label{Tdef}
(T_\pi \psi)(x_1,...,x_N):=\psi(x_{\pi^{-1}(1)},
x_{\pi^{-1}(2)},...,x_{\pi^{-1}(N)}).
\end{equation}
We label by $\sigma$ the types of the irreducible representations of
$T$. Let $\mathcal{H}^{\sigma}$ be the subspace of
$L^2(\mathbb{R}^{3N})$ on which the representation $T$ is a multiple
to the irreducible representation of type $\sigma$. Then
\begin{equation}\label{l2decomposition}
L^2(\mathbb{R}^{3N})=\oplus_\sigma \mathcal{H}^{\sigma}.
\end{equation}}
\DETAILS{Recall the definitions concerning irreducible  representations of the permutation group $S_N$ of $\{1,2,...,N\}$, 
given in  the introduction.}

Recall from the introduction that  $\mathcal{H}^{\sigma}$ denotes the subspace of
$L^2(\mathbb{R}^{3N})$ on which the representation  of  the permutation group $S_N$  is a multiple
to the irreducible representation of type $\sigma$, $H^\sigma$,  the restriction  of $H$ onto $\mathcal{H}^{\sigma}$ and $E^\s(y)= \inf \sigma(H^\s)$, the ground state energy of the system for the irreducible
representation $\s$.

Let $\chi_g^{\sigma}$ denote the character of the representation of type $\sigma$ evaluated at $g$. 
The orthogonal projection of $L^2(\mathbb{R}^{3N})$ onto
$\mathcal{H}^{\sigma}$  is given by (see \cite{Ha})
\begin{equation}\label{Psigma}
Q^\sigma=\frac{\chi_{id}^{\sigma}}{N!} \sum_{\pi \in S_N}
\chi_{\pi^{-1}}^{\sigma} T_\pi.
\end{equation}
Since $H$ commutes with any $T_\pi$, by
\eqref{Psigma}, so it does with $Q^{\sigma}$,
\begin{equation}\label{PsigmaHcom}
H Q^{\sigma}=Q^{\sigma} H.
\end{equation}

\DETAILS{Hence $\mathcal{H}^{\sigma}$ is an invariant space of $H$. Let
$H^\sigma$ be the restriction of $H$ onto $\mathcal{H}^{\sigma}$. We
now denote the ground state energy of the system for the irreducible
representation $\s$ by
\begin{equation*}
E^\s(y)= \inf \sigma(H^\s).
\end{equation*}
In what follows we fix an irreducible representation 
of the group $S_N$, which we denote by $\sigma \equiv \s (S_N)$.

We repeat some definitions from the introduction. Let $S(a)$ be the subgroup of the permutation group $S_N$
that keeps the clusters of the decomposition $a$ invariant. {\bf The space of $\sigma$ is of course invariant under $T_{\pi},\ \pi \in S(a)$, but  the restriction $\s |_{S(a)}$ of the representation $\sigma$ of $S_N$ is not necessarily irreducible. Denote by $I^\s$ the
 family of irreducible representations $\al\equiv \al (S(a))$ of $S(a)$ such that
\begin{equation*}
\sigma|_{S(a)}=\oplus_{\alpha \in I^\s} \alpha,
\end{equation*}
in the sense of the corresponding subspaces.}
 The representations $\alpha \in I^\s$ are called induced
representations and we write $\alpha \prec \sigma$, so that $ I^\s=\{\al\equiv \al (S(a)) : \alpha \prec \sigma\}$. 

In what follows we denote irreducible representations of the group $S(a)$ by $\al, \beta, \g$. The corresponding decomposition $a$ will always be clear from the context.}

We also recall from  the introduction the definition of the subgroups $S(a)$ of the permutation group $S_N$, which leave the decompositions $a$ invariant, with its  irreducible  representations  denoted in what follows by $\al, \beta, \g$ (the corresponding decomposition $a$ will always be clear from the context), the notion of the  irreducible representation  
and the notation $\alpha \prec \sigma$ specifying that  an irreducible representation $\alpha$ of $S(a)$ is  induced  by an irreducible  representation $\sigma$ of $S_N$.
\DETAILS{ and the corresponding decomposition of the Hilbert space $\cH^\sigma$, as $\cH^\sigma=\oplus_{\alpha  \prec \sigma} \cH_a^\alpha$ (which we denote also as $\sigma|_{S(a)}=\oplus_{\alpha \prec \s} \alpha$, see \eqref{sig-deco}). }
Now, we denote by $Q_a^{\alpha}$ the orthogonal projection onto the
subspace of $L^2(\R^{3N})$ on which the representation of $S(a)$ is
a multiple to the irreducible representation of the type $\al$. By \cite{Ha}, it can be written as
\begin{equation}\label{Pbet}
Q_a^{\alpha} =\frac{\chi_{id}^\alpha}{\#S(a)}\sum_{g \in S(a)}
\chi_{g^{-1}}^{\alpha} T_g,\end{equation} where $\#S(a)$ is the
cardinality of $S(a)$. By the 
the definition of the induced representations, 
\begin{equation}\label{Qsig-deco} Q^\s= \sum_{\al \prec \sigma}Q_a^\al Q^\s=  Q^\s\sum_{\al \prec \sigma}Q_a^\al.\end{equation}
Since $H_a$ commutes with all permutations in $S(a)$, it commutes with $Q_a^{\alpha}$.   
 Recall that the restriction of $H_a$ onto $ 
  \Ran Q_a^{\alpha}$ is denoted by $H_a^{\alpha}$.
     In what follows we fix an irreducible representation,  $\sigma \equiv \s (S_N)$, 
of $S_N$. 
 \DETAILS{   We denote by $H_a^{\alpha}$ the restriction of $H_a$ onto $ \Ran Q_a^{\alpha}$.
 We write $\alpha \prec \prec \sigma$ if $\alpha \prec \sigma$ and
$\inf\s(H_a^{\alpha} )= \min_{b \in \mathcal{A}^{at},\beta \prec
\sigma}\inf\s(H_b^\beta).$ Let
\begin{equation}\label{Esinftydef}
E^{\s}(\infty):=\min_{a \in \mathcal{A}^{at}, \al\prec \prec \s}
\inf \sigma(H_{a}^{\alpha}).
\end{equation}}

We now choose the orthogonal projection for the Feshbach map. 
For given  $a \in \mathcal{A}^{\rm at}$ and $\alpha \prec \sigma$,  let
$P_a^\al$ be the orthogonal projection onto the ground state eigenspace of $H_a^\al$. 
 Let $\chi_R: \mathbb{R}^3 \rightarrow
\mathbb{R}$ be a spherically symmetric smoothed characteristic
function of the ball of radius $\frac{R}{8}$ around the nucleus and
supported in this ball. 
 Let \begin{align}\label{PaR} P_{a,R}^{\alpha}\   \mbox{be the orthogonal projections onto 
the vector spaces}\  \prod_{A\in a}\chi_R^{\otimes |A|}  \Ran P_{a}^{\alpha}.\end{align}  
By the properties of the cut-offs and decompositions, we see that 
$\Ran P_{a, R}^\alpha,  a \in \mathcal{A}^{\rm at}, \alpha \prec  \sigma,$ consist of functions of mutually disjoint supports. 
 Hence   
the projections $P_{a, R}^\alpha,  a \in \mathcal{A}^{\rm at}, \alpha \prec \prec  \sigma,$ are mutually orthogonal. Hence we can define the orthogonal projection 
\begin{align}\label{Psig} 
P^\s =\sum_{a \in \mathcal{A}^{\rm at}, \alpha \prec \prec \sigma} P_{a, R}^\alpha =\sum_{a \in \mathcal{A}^{\rm at}} P_{a, R}.
\end{align}
%
\DETAILS{\begin{equation}\label{abgeschittenesraum}
P^\s=\ \mbox{the orthogonal projection onto the subspace}\ \text{span} \{\Ran P_{a,R}^{\alpha} | a \in \mathcal{A}^{at}, \alpha \prec \prec \sigma\}.
\end{equation}  where the latter symbol is introduced in Introduction,}
  where $P_{a, R} :=\sum_{\alpha \prec \prec \sigma} P_{a, R}^\alpha$. Properties of $P^\sigma$  are described in two lemmas below. 
\begin{lemma}\label{lem:QcommT}
$P^\sigma$ commutes with $T_{\pi}$, for any $\pi \in S_N$, and
therefore with  $Q^{\sigma}$.
\end{lemma}
\begin{proof}
 To show that $T_{\pi}$ commutes with $P^\sigma$, 
 we observe that
\begin{equation*}
 T_{\pi} H_a T_{\pi^{-1}}=H_{\pi a},\ \quad \mbox{and therefore}\ \quad
 T_{\pi} P^\al_{a, R} T_{\pi^{-1}}=P^\al_{\pi a, R},
\end{equation*}
 which, together with \eqref{Psig},  gives the desired statement. 
 \DETAILS{if $\Psi \in \Ran P^\sigma$, then $T_{\pi} \Psi \in \Ran P^\s$. To show
the latter it suffices to take $\Psi$ satisfying
 \begin{equation*}
H_a \Psi=E^{\sigma}(\infty) \Psi.
\end{equation*}
Therefore, $(T_{\pi} H_a T_{\pi^{-1}}) T_{\pi} \Psi=E^{\sigma}(\infty) T_{\pi}\Psi.$ In addition we have that
\begin{equation*}
 T_{\pi} H_a T_{\pi^{-1}}=H_{\pi a},
\end{equation*}
  where the clusters of $b=\pi a$ were defined in
  \eqref{decopermu1111}. The last two equations imply that
\begin{equation*}
H_{\pi a} T_{\pi} \Psi =E^{\sigma}(\infty) T_{\pi} \Psi.
\end{equation*}
As a consequence, $T_{\pi} \Psi \in \Ran P^\s$.}
\end{proof}
\begin{lemma}\label{lem:RanPsigQsig}  
For each $\alpha \prec \prec \sigma$, any normalized $\Psi\in \Ran P_{a,R}^{\alpha}$ and for $R$ is large enough,  we have the estimate
\begin{equation}\label{Psigmaphiaalphanorm}
\|Q^\sigma \Psi\|^2 = \frac{\chi_{id}^{\sigma}}{N!} \frac{\#
S(a)}{\chi_{id}^{\alpha}},
\end{equation}
 and therefore $\Ran P^\s \cap \Ran Q^\s \neq \{0\}$.
\end{lemma}
\begin{proof}
Due to the cut off of the ground states of the atoms, we have that
$\langle \Psi, T_{\pi} \Psi \rangle = 0, \forall \pi \in S_N/S(a).$
Therefore, using the expression for $Q^\sigma$, we have that
\begin{align}\label{Psigmaphi1}
\|Q^\sigma \Psi\|^2 = \langle \Psi, Q^\sigma \Psi \rangle &=
\frac{\chi_{id}^{\sigma}}{N!} \sum_{\pi \in S_N}
\chi_{\pi^{-1}}^{\sigma} \langle  \Psi, T_{\pi} \Psi \rangle\notag \\ & =
\frac{\chi_{id}^{\sigma}}{N!} \sum_{\pi \in S(a)}
\chi_{\pi^{-1}}^{\sigma} \langle  \Psi, T_{\pi} \Psi \rangle.
\end{align}
  Using $\s|_{S(a)}=\oplus_{\al \prec \sigma} \al$ and  taking the trace of $T_{\pi^{-1}}^\s=\oplus_{\al
\prec \sigma} T_{\pi^{-1}}^\al,\ \forall \pi\in S(a)$,  we obtain
\begin{equation}\label{char-rel}
\chi_{\pi^{-1}}^\sigma = \sum_{\beta \prec \sigma}
\chi_{\pi^{-1}}^{\beta},\ \quad \forall \pi \in S(a).
\end{equation}
This, together with the expression \eqref{Pbet} for the projection
$Q^\beta_{S(a)}$ and \eqref{Psigmaphi1} gives that 
\begin{equation*}
\|Q^\sigma \Psi\|^2 = \frac{\chi_{id}^{\sigma}}{N!}\langle \Psi,
\sum_{\beta \prec \sigma} \frac{\# S(a)}{\chi_{id}^{\beta}} Q^\beta
\Psi \rangle.
\end{equation*}
The last equation together with the fact that  $Q_a^\beta
\Psi=\delta_{\alpha,\beta} \Psi$ gives \eqref{Psigmaphiaalphanorm}.
\end{proof}
By Lemmas \ref{lem:QcommT} and \ref{lem:RanPsigQsig}, the
operator
\begin{equation}\label{Pstat}
\Pi^\s:=P^\s Q^\s=Q^\s P^\s,
\end{equation}
 is the orthogonal projection on $\Ran
P^\s \cap \Ran Q^\s\ne \varnothing$ and satisfies
\begin{equation}\label{PsigmaPgeneral}
 \Pi^\s Q^\s=Q^\s \Pi^\s=\Pi^\s.
\end{equation}
 We now consider the Feshbach map,  $F_{\Pi^\s}(\lambda)$,  where $\Pi^\s$ is defined in \eqref{Pstat}, as defined in \eqref{FP}  - \eqref{U}, i.e.
 \begin{equation}\label{FPsigma}
F_{\Pi^\s}(\lambda)= \big[ \Pi^\s H \Pi^\s- \Pi^\s H
\Pi^{\s\bot}(\Pi^{\s \bot} H \Pi^{\s \bot} -\lambda)^{-1}
\Pi^{\s\bot} H \Pi^\s \big]|_{\Ran \Pi^\s},
\end{equation}
 with, as above,  $\Pi^{\s\bot}=Q^\s (1-\Pi^\s) $.
 Since $H$ and $P^\s$ commute with $Q^\s$, we have that
\begin{equation}\label{FQsigma}
F_{\Pi^\s}(\lambda)= \left[ P^\s H^{\sigma}
P^\s-U^\s(\lambda)\right]\big|_{\Ran \Pi^\s},
\end{equation}
where
\begin{equation}\label{Usigmaldef}
U^\s(\lambda)=P^\s H^\s P^{\s\bot} (H^{\sigma \bot}-\lambda)^{-1}
P^{\s\bot} H^\s P^\s,
\end{equation}
 with $P^{\s\bot}=1-P^\s$ and $H^{\s \bot}=P^{\s\bot} H^\s P^{\s\bot}$.
\DETAILS{\begin{equation}\label{FPsigma} F_P(\lambda)= \left[ P H
P-U(\lambda)\right]\big|_{\Ran P},
\end{equation}
with
\begin{equation}\label{UsigmaP}
U(\lambda)=P H P^\bot(H^{ \bot}-\lambda)^{-1} P^\bot H P,
\end{equation}
 and $P^\bot=\one-P$ and $H^{ \bot}=P^\bot H P^\bot$.}
  As before (see \eqref{FSE}) we have
\begin{equation}\label{FSE-sig}
\lambda \text{ eigenvalue of } \H^{\sigma} \iff \lambda\ \text{
eigenvalue of }\ F_{\Pi^\s}(\lambda)
\end{equation}
and to prove Theorem \ref{thm:vdW-sig} we have to estimate the
different terms on the r.h.s. of \eqref{FQsigma}.

 Finally, we note that any irreducible  representation $\alpha$ of $S(a)$ is a product of irreducible representations $\alpha_j$  of the groups $S(A_j)$ of permutations of the clusters $A_j\in a$.  We write symbolically,  $\alpha=\otimes_{j=1}^M \alpha_j$. The corresponding factorization of $Q_a^\alpha$ is given in (see Appendix \ref{sec:sym-gen-app})
\begin{equation}\label{Qaalpha-factor}
Q_a^\alpha=\prod_{j=1}^M Q_{A_j}^{\alpha_j}.
\end{equation}
Keeping in mind that  $\alpha=\otimes_{j=1}^M \alpha_j$ determines uniquely $\al_j$, we write in what follows $Q_{A}^\al$ for $Q_{A}^{\al_j}$.  Similarly, we denote by  $P_{A}^\al$ and $P_{A, R}^\al$  the orthogonal projections onto the ground state eigenspace and  the cut-off ground state eigenspace  of an atom $A$.  The relation  \eqref{Qaalpha-factor}   implies
\begin{equation}\label{Paal-factor} P_a^\al= \otimes_{A\in a} P_{A}^\al\ \quad \mbox{and}\  \quad  P_{a, R}^\al= \otimes_{A\in a} P_{A, R}^\al.
\end{equation}


\paragraph{Estimate of $P^\s H^\s P^\s$.}\label{sec:php}

\begin{lemma}\label{lem:PHP}
\begin{equation}\label{PHPstat}
P^\s H^\s P^\s \doteq E^\s(\infty) P^\s.
\end{equation}
\end{lemma}
\begin{proof}
In this proof we omit the superindex $\sigma$ in $P^\s$ 
and $E^\s(\infty)$ and write instead $P$  
and $E(\infty)$.  
Eq. \eqref{Psig} and the relations $H=H_a+I_a$ and  $H_a  P_{a, R}^\alpha \doteq E(\infty)  P_{a, R}^\alpha,\ \forall a \in \mathcal{A}^{at}, \alpha \prec \prec \sigma,$  we obtain that $PH^\s P \doteq E(\infty) PQ+\sum_{a, b, \al,\beta} P_{a, R}^\alpha I_aP_{b, R}^{\beta}$.  
Since,  for $a \neq b$, $\Ran  P_{a, R}^\alpha$ and $\Ran  P_{b, R}^{\beta}$ have disjoint supports, we conclude that
\begin{equation}\label{PHPsigmasigma}
PH^\s P \doteq E(\infty) PQ+
\sum_{a, \al,\beta} P_{a, R}^\alpha I_aP_{a, R}^{\beta}.
\end{equation}
We will now show that
\begin{equation}\label{PaIaPa-stat} 
 P_{a, R}^\alpha I_aP_{a, R}^{\beta}=0,\ \forall a, \al,\beta.
\end{equation}
 We first consider the case  $\alpha \neq \beta$. Since $I_a$ commutes with $T_\pi$ for all
permutations in $ \pi \in S(a)$ and therefore, due to \eqref{Pbet},
with $Q_a^\beta$, we have that $I_a  \Ran P_{a, R}^{\beta} \subset \Ran
Q_a^{\beta}$. Since also $ \Ran P_{a, R}^{\alpha} \subset \Ran Q_a^{\alpha}$,
it is orthogonal to $I_a  \Ran P_{a, R}^{\beta},\ \beta\ne \al$, and therefore \eqref{PaIaPa-stat} holds for $\alpha \neq \beta$.

We now consider the  case $\alpha = \beta$. 
Clearly, the map $P_{a, R}^\alpha I_aP_{a, R}^{\alpha}\big|_{\Ran P_{a, R}^{\alpha}}$ leaves the space $\Ran P_{a, R}^{\al}$ invariant and  commutes with $T_{\pi},\ \forall \pi\in S(a)$.  Since by Condition (D) $\Ran P_{a, R}^{\alpha}$ is a space of an  irreducible  representation of $S(a)$, we conclude that  it is a multiple of the identity, $P_{a, R}^\alpha I_a P_{a, R}^{\alpha}\big|_{\Ran P_{a, R}^{\alpha}}=\lam \one$ for some real $\lam$. Hence  $P_{a, R}^\alpha I_a P_{a, R}^{\alpha}=\frac{1}{\rank(P_{a, R}^{\alpha})}\Tr (I_a P_{a, R}^{\alpha}) P_{a, R}^{\alpha}$, where $\rank(P_{a, R}^{\alpha})$ is the rank of $P_{a, R}^{\alpha}$. 

 Using  the definition of $\rho_{A}^\al$ ({\bf see} \eqref{rhoA-def}) and the factorization  $P_{a, R}^\al= \otimes_{A\in a} P_{A, R}^\al$, described at the end of the last subsection,  and proceeding as in Lemma \ref{lem:newton}, we arrive at  \eqref{PHPstat}. 
\DETAILS{by the decomposition $\Tr (I_a P_{a, R}^{\alpha}) =\sum_{i<j}^{1,M} 
 \sum_{k \in A_i, l \in A_j}\Tr (I_{ij}^{kl} P_{a, R}^{\alpha})$ (see \eqref{Iabreakup})  we have, after changing the variables of integration according to  \eqref{z-var} {\bf (see also \eqref{tildeIijkl})}, 
 \begin{equation}\label{TrIP}\Tr (I_a P_{a, R}^{\alpha}) = \sum_{i<j}^{1,M} \sum_{k \in A_i, l \in A_j}\int\int \tilde I_{ij}^{kl} (z, z')\rho_{A_i}^\al(z)\rho_{A_j}^\al(z') dz dz' .\end{equation}
By 
Proposition \ref{prop:spherical}, $\rho_{A}^\al$ is spherically symmetric. This implies $\Tr (I_a P_{a, R}^{\alpha}) =0$, which, in turn, gives    \eqref{PaIaPa-stat} for $\alpha = \beta$. 
Hence \eqref{PaIaPa-stat} and therefore \eqref{PHPstat} follow.}   
\end{proof}
Before proceeding to estimating $U^\s:=U^\s(E^\s)$, we present some preliminary results.
 In Appendix \ref{sec:Hbot-low-bnd-stat}, we prove that 
there exists $\gamma^\sigma>0$ and $C>0$ such that
\begin{equation}\label{Hbot-low-bnd-stat}
H^{ \sigma \bot} \geq (E^{\s}(\infty)+ 2 \imsgap^{\sigma}-\frac{C}{R}) Q^\s.
\end{equation}
Furthermore,  similarly to \eqref{ineqE}, Eq. \eqref{PHPstat} implies the following elementary variational estimate
\begin{equation} \label{Es-rough-es}
E^\s(y)\  \dot \leq\  E^\s(\infty). 
\end{equation}

\paragraph{Estimate of $U^\s:=U^\s(E^\s)$.}\label{sec:Ustatest}

Denote $P_{ i j}^{\al}:= P_{ A_i}^{\al}\otimes P_{A_j}^{\al}$ and $P_{ij}^\bot := Q_{A_i}^{\al}\otimes Q_{A_j}^{\al} - P_{ A_i}^{\al}\otimes P_{A_j}^{\al}$ (see \eqref{Qaalpha-factor} and \eqref{Paal-factor} for the definition of $Q_{A}^{\al}$ and $P_{A}^{\al}$).
 Let also 
 \begin{equation}\label{rkldef}
R_{k l}^\bot := (P_{kl}^\bot (H_{A_k}+H_{A_l}) P_{kl}^\bot-E_{A_k}^\s-E_{A_l}^\s)^{-1}
\end{equation} 
 (not to be confused with the related object introduced after \eqref{Habot} and denoted by the same letter). Our goal is to prove the following lemma:
\begin{lemma}\label{lem:Usigmaapprox2}
The equation  \eqref{Usigmaapprox2} holds
 \begin{equation}\label{Usigmaapprox2}
U^\s= \sum_{\alpha \prec  \prec \sigma} P^{\alpha} \sum_{i<j}^{1,M}
\frac{\sigma_{ij}^{\s,\alpha}}{|y_i-y_j|^6}+O(\frac{1}{R^7}), 
 \end{equation}
 where $\sigma_{ij}^{\sigma, \alpha}$ 
 are positive and independent of $y$, given  by
 \begin{equation}\label{sijsalpha}
\sigma_{ij}^{\sigma, \alpha}:= 
\Tr ( f_{ij } P_{ i j}^{\al} R_{ij}^\bot  f_{ij} P_{  i j}^{\al}), \end{equation}  
with recall, $f_{ij}$ defined in \eqref{fij}, and
  \begin{equation}\label{Palpha'}
 P^\alpha:=Q^\sigma \sum_{a \in \mathcal{A}^{at}}  P_a^\alpha Q^\sigma.
\end{equation}
\end{lemma}
\DETAILS{\noindent {\bf Remark.} The proof below is different from the one in Section \ref{sec:setup}. Instead of weighted estimates we use a geometrical decomposition of the operator $H^{\s \bot}$ proven above. This decomposition is similar to \eqref{Hbot-low-bnd-stat2} but has equality instead of the inequality.}
\begin{proof}
 By \eqref{Hbot-low-bnd-stat} and  \eqref{Es-rough-es}, the operator $H^{\s \bot}-E^\s$,   where, recall,  $H^{\s \bot}=P^{\s\bot} H^\s P^{\s\bot}$, has a bounded inverse, which we denote by 
 $R^{\s \bot}(E):=(H^{\s \bot}-E )^{-1}$.

In this proof, we will omit for simplicity the superindex $\sigma$ in $
P^\s$, $Q^\s$, $E^\s$, $E^\s(\infty)$, $\gamma^\s$, $\s_{ij}^{\s,\al}$, $R^{\s \bot}(E)$ and write, instead, $ P$, $Q$, $E$, $E(\infty)$, $\gamma$, $\s_{ij}^{\al}$, $R^{\bot}(E)$ and do not specify the exact range $a,b \in \mathcal{A}^{at},\alpha, \beta \prec \prec
\sigma$ in summations.
\DETAILS{ Recall that $U^\s(\lambda)$ was defined in \eqref{Usigmaldef}.
Using that $H^\s=H Q$ and the fact that
 $Q$ commutes with $H$ and $P$, we obtain that
\begin{equation}\label{Usfirst}
U^\s= Q  P H P^\bot (H^{\s \bot} -E )^{-1} P^\bot H P Q .
\end{equation}
where, recall,  $H^{\s \bot}=P^{\s\bot} H^\s P^{\s\bot}$. }
   Using  equations \eqref{Psig} and  $H_a P_{a, R}^\alpha \doteq \Einfty P_{a, R}^\alpha$ and $P^\bot P_{a, R}^{\al}=0$, we obtain that
 \begin{equation}\label{PperpHP-exp} P^\bot H P \doteq \sum_{a,\al}  P^\bot I_a P_{a, R}^{\al}.
 \end{equation}
 (see \eqref{Paal-factor} for the definition of $P_{A, R}^{\al}$.)
 Now, using the definition of  $U^\s(\lambda)$ in \eqref{Usigmaldef},  \eqref{PperpHP-exp} and the fact that $Q$ commutes with $H$ and $P$, we obtain that 
  \begin{align}\label{Usig-aprox1}
U^\s (E) \dot = \sum_{a,b,\alpha, \beta} 
 P_{a, R}^{\al}  I_a   P^\bot R^{\bot}(E) P^\bot  I_b P_{b, R}^{\beta}. 
\end{align}
 Proceeding as in the proof of \eqref{IaPsia-est} we can obtain
 that
 \begin{equation}\label{IaPa-est-stat}
\|I_a P_{a, R}^{\al}\| \lesssim \sum_{i<j} \frac{1}{|y_i-y_j|^3},
\forall a \in \mathcal{A}^{at}. 
 \end{equation}

Furthermore, recall the functions $\chi_a$,  defined in \eqref{chiader'} -  \eqref{chia-supp} and satisfying  
\begin{equation}\label{chiaIa-sig} I_a (-\Delta +1)^{-1/2}=O(\frac1R) \text{ on } \supp \chi_a.\end{equation}
\begin{equation}\label{chiaP-sig}
\chi_a P_{b, R}  = \delta_{ab}  P_{a, R}, \quad 
  [\chi_a,  P_{a, R}^{\al \bot}]=0,\ 
  \chi_a P^{\bot}  =  P_{a, R}^{\bot}\chi_a, \end{equation} 
$\forall a,b \in \mathcal{A}^{at}$. Now, let $ R_a^{ \bot}(E):=(H_a^{\bot}   -E)^{-1} $,   with $ H_a^{ \bot}:=\sum_{\alpha \prec \prec \sigma} H_a^{\al}  P_a^{\al \bot}$  (this is the generalization of the operator denoted by the same symbol in \eqref{charcom}). Similarly as in \eqref{charcom}, we show
\begin{equation}\label{charcom-sig}
\chi_a R^{ \bot}(E)- R_a^{ \bot}(E)\chi_a =O(\frac{1}{R}).
\end{equation}
 
 Since $\chi_a$ commutes with $I_a$, we have $ P_{a, R}^{\al}  I_a    = P_{a, R}^{\al}  I_a   \chi_a $. Using this to insert $\chi_a$ into \eqref{Usig-aprox1},  and using \eqref{chiaP-sig}, \eqref{charcom-sig} and \eqref{IaPa-est-stat} 
 gives 
 \begin{align} \label{Usig-aprox2} 
 U^\s (E) &= \sum_{a,b,\alpha, \beta} 
 P_{a, R}^{\al} (E) I_a     R_a^{ \bot}(E)\chi_a  I_b P_{b, R}^{\beta}  +O(\frac{1}{R^7}))\notag \\ 
&= \sum_{a,\alpha, \beta}U_{aa}^{\alpha \beta}(E) +O(\frac{1}{R^7}),     
\end{align}
where $U_{aa}^{\alpha \beta}(E)  =P_{a, R}^{\al} I_a   R_a^{\bot}(E) I_a P_{a, R}^{\beta}$.

Next, as above, using the orthogonality to the subspaces corresponding to different irreducible representations of $S(a)$,  we obtain 
\begin{align}\label{Uab4-sig}
 U_{aa}^{\alpha \beta}(E) & =P_{a, R}^{\al} I_a   R_a^{ \bot}(E) I_a P_{a, R}^{\al }  \delta_{\al \beta}  +O(\frac{1}{R^7}). \end{align}
 Next,  
 as in going from \eqref{U-aprox2} to \eqref{Uab}, we pass from $R_a^{ \bot}(E)$ to $R_a^{ \bot}(E(\infty))=: R_a^{ \bot}$, to obtain
 \begin{align}\label{Uab-sig}
U^\s(E)&= \sum_{a,\alpha, \beta}U_{aa}^{\alpha \beta} +O(\frac{1}{R^7}),\  U_{aa}^{\alpha \beta}  =P_{a, R}^{\al} I_a   R_a^{\bot} I_a P_{a, R}^{\beta}  . 
\end{align}

Now, as discussed in Subsection \ref{sec:deco}, for any $a,b \in \mathcal{A}^{at}$
there exists a permutation $\pi$ such that $b=\pi a$. Since
on the other hand $T_\pi$ is unitary and commutes with $Q, P, H$ and since $P_{a, R}^{\al}T_\pi^{-1} =T_\pi^{-1} P_{b, R}^{\al} $, where $b=\pi a$, we obtain that
\begin{align*}
 P_{a, R}^{\al}T_\pi^{-1} & H  P^\bot R^{ \bot}(E) P^\bot H T_\pi P_{a, R}^{\al}=T_\pi^{-1} P_{b, R}^{\al}  H  P^\bot R^{ \bot}(E) P^\bot H  P_{b, R}^{\al}T_\pi,\end{align*}
Due to the definition $U_{aa}^{\alpha \beta}$ in \eqref{Uab-sig}, 
we have 
  \begin{align}\label{Uab5-sig}
&U_{aa}^{\alpha \alpha}=  T_\pi^{-1} U_{bb}^{\alpha \alpha} T_\pi,\  \mbox{with}\ b= \pi a,  
\quad \forall a \in \mathcal{A}^{at}.
\end{align}

Now, we use again, as in the proof of Lemma \ref{lem:PHP}, that since $U_{aa}^{\alpha \alpha}\big|_{\Ran P_{a, R}^{\al}}$ leaves the space $\Ran P_{a, R}^{\al}$ invariant and  commutes with the irreducible representation $T_{\pi}^{\alpha}$, it is a multiple of identity.  
This gives 
$U_{aa}^{\alpha \alpha}=\mbox{multiple of}\ P_{a, R}^{\al} . $ 
This implies
 \begin{align}\label{Vaaalal1} U_{aa}^{\alpha \alpha} = \frac{1}{\rank P_{a, R}^{\al}}\Tr (U_{aa}^{\alpha \alpha} P_{a, R}^{\al}) P_{a, R}^{\al}. \end{align}

 Since $I_a {P_{a, R}^\al} \in \text{Ran} Q_a^\al$, the summands in  $R_a^{ \bot}I_a {P_{a, R}^\al}=(\sum_{\beta \prec \prec \sigma}  H_a^{\beta \bot}  -E(\infty))^{-1} I_a {P_{a, R}^\al}$, 
 with $\beta \neq \alpha$, vanish.  Moreover, due to the exponential decay of the ground states and their derivatives up to second order we can replace -  with only an exponentially small error - $P_{a, R}^{\al \bot}$ in the resulting term $(H_a^{\al \bot} -E(\infty))^{-1} $  
by $P_a^{\al \bot}$.
  Therefore,  
\begin{align}\label{Vaaalal2}
 U_{aa}^{\alpha \alpha} = P_{a, R}^{\al}  I_a  R_a^{\al \bot}  I_a P_{a, R}^{\al} 
+O(\frac{e^4}{R^7}), 
\end{align}
where $ R_a^{\al \bot}:=( H_a^{\alpha \bot}- \Einfty )^{-1}$. 

Now using \eqref{Vaaalal1}  and \eqref{Vaaalal2} and  the formula \eqref{Iabreakup} (
$I_{a}=\sum_{i<j}^{1,M} 
\sum_{k \in A_i, l \in A_j}I_{ij}^{kl}$, where ) and the equation,  similar to the equation \eqref{Iijkl-exp}, 
\begin{equation}\label{Iijkl-exp-P} 
I_{ij}^{kl}P_{a, R}^{\al}=\frac{e^2}{|y_{ij}|^3}f_{ij}^{lm}(z,\widehat{y_{ij}})P_{a, R}^{\al}+O(\frac{1}{|y_{ij}|^4}),\end{equation}
where 
$\widehat{y}_{ij}=\frac{y_{ij}}{|y_{ij}|}$,  $z:=(z_{ki}, z_{lj},\ k \in A_i,\  l  \in A_j)$ with the variables $z_{ki}, z_{lj}$, defined in \eqref{z-var}, and $f_{ij}(z,\widehat{y}_{ij})$ are given  after \eqref{Iijkl-exp}, 
we obtain
\begin{align}\label{Vaaalal3}
 \Tr (U_{aa}^{\alpha \alpha} P_{a, R}^{\al}) = \sum_{i<j}^{1,M} \sum_{k<l}^{1,M}  \frac{ e^4 W_{ij;kl}}{|y_i-y_j|^3 |y_k-y_l|^3} 
 +O(\frac{e^4}{R^7}), 
\end{align}
where 
\begin{equation}\label{W}
W_{ij;kl}:=\Tr ( f_{ij }  R_a^{\al \bot}  f_{k l} P_{a, R}^{\al}).
\end{equation}
\DETAILS{We will show that
\begin{equation}\label{W0}
W_{ij;kl}=0, \text{ when } ij \neq kl.
\end{equation}
{\bf (compare the rest of the sect with the corresponding parts of Sect \ref{sec:setup})} Indeed, assume without loss of generality that $i \neq k,l$. Then by
\eqref{fAiAj} $f_{ij}$ is odd in the variables $z_{mi}, m \in A_i$
defined in \eqref{zlm}. One the other hand the 
the product of the remaining factors in \eqref{W} has even electron densities in these variables (because $P_{kl}, H_{A_k}, H_{A_l}, f_{kl}$ act
on different variables and because of Lemma \ref{lem:rho1}). Hence \eqref{W0} follows.}

As in \eqref {W-expr}, we show that $W_{ij;kl}=\s_{ij}^{\al}\del_{ij;kl}$. The part $ij \neq kl$ is obtained in exactly the same way. For   $ij =kl$,
we use  the factorization \eqref{Paal-factor}  of $P_{a, R}^{\al}$,  to obtain that
\begin{align}\label{decouplestat}
& R_a^{\al \bot} f_{kl} P_{a, R}^{\al}=  R_{kl}^\bot f_{kl} P_{a, R}^{\al}= \prod_{m\ne i, j} P_{A_m, R}^{\al} R_{i j}^\bot f_{ij} P_{i j, R},
\end{align}
where, recall, the operators $R_{kl}^\bot $ are given by \eqref{rkldef}.  Denote $P_{ i j, R}^{\al}:= P_{ A_i, R}^{\al}\otimes P_{A_j, R}^{\al}$ and $P_{ij}^\bot := Q_{A_i}^{\al}\otimes Q_{A_j}^{\al} - P_{ A_i, R}^{\al}\otimes P_{A_j, R}^{\al}$  (see \eqref{Qaalpha-factor}  for the definition of $Q_{A}^{\al}$).	Inserting this into \eqref{W}, with  $ij =kl$, and passing from $P_{i j, R}$ to $P_{i j}$, gives
\DETAILS{This and the factorization of $P_{a, R}^{\al}$ (see \eqref{Paal-factor}) implies 
\begin{equation}\label{W2}
W_{ij;kl}:=\Tr ( f_{ij } R_{kl}^\bot  f_{k l} P_{i j A_k A_l, R}^{\al}).
\end{equation}
Using that  \eqref{Paal-factor} implies that $ R_{ij}^\bot f_{ij} P_{a, R}^{\al}= \prod_{m\ne i, j} P_{A_m, R}^{\al} R_{ij}^\bot f_{ij} P_{A_i A_j, R}$, we have, for   $ij =kl$,}
\begin{equation} 
W_{ij;ij}=\Tr ( f_{ij } P_{i j, R}^{\al} R_{ij}^\bot  f_{ij} P_{ i j, R}^{\al})\dot =\Tr ( f_{ij } P_{i j}^{\al} R_{ij}^\bot  f_{ij} P_{ i j}^{\al})=: \s_{ij}^{\al}.\end{equation}
This shows  $W_{ij;kl}=\s_{ij}^{\al}\del_{ij;kl}$, which, together with \eqref{Uab-sig},  \eqref{Uab4-sig},  \eqref{Vaaalal1}, and \eqref{Vaaalal3}, 
 implies the relation 
\eqref{Usigmaapprox2}-\eqref{Palpha'}.
Finally, the proof that  $\sigma_{ij}^{\al}$ are  positive and independent of   $y$ is done similarly  as for the case without statistics.
\DETAILS{ as for the case without statistics, 
since $R_{ij}^\bot$ is positive, so is $\sigma_{kl}^{\s,\al}$. 
To show the independence of  $\sigma_{ij}$ of $y$, we first note that the only
dependence of $\sigma_{ij}$ on $\widehat{y_{ij}}:=\frac{y_i-y_j}{|y_i-y_j|}$
 appears on $f_{ij}$ so we will write $f_{ij}^{\widehat{y_{ij}}}$ and
$\sigma_{ij}^{\widehat{y_{ij}}}$. For any rotation $R$ in
$\mathbb{R}^3$ we define $T_R$ acting on the space
$L^2(\mathbb{R}^{3(|A_i|+|A_j|)})$ (recall that $|A_k|$ is the
number of electrons of the $k$-th atom) of functions of the
variables $z_{kl}=x_k-y_l,\ k=1, \dots , |A_i|+|A_j| ,\ l =1, 2$, by
rotating them.
 We have proven in the proof of Proposition \ref{spherical} that
 \begin{equation}\label{statTRphiij}
 T_R P_{ A_i A_j, R}^{\al}T_R^{-1} =c(R) P_{ A_i A_j, R}^{\al} \text{ where } |c(R)|=1.
 \end{equation}
 On the other hand using \eqref{fAiAj} and the fact that $R$ is
 unitary we obtain that
\begin{equation}\label{statfijrotation}
T_R^{-1}f_{ij}^{\widehat{y_{ij}}}=f_{ij}^{\widehat{R^{-1} y_{ij}}}.
\end{equation}
       Using \eqref{statTRphiij}, \eqref{fijrotation} and the fact that $T_R$
commutes with the Hamiltonians $H_i$, $H_j$ we can obtain that
$\sigma_{ij}^{\widehat{y_{ij}}}=\sigma_{ij}^{\widehat{R y_{ij}}}$
implying the independence of $\sigma_{ij}$ of $y$. }
%
  \end{proof}

\paragraph{Completion of the proof of Theorem \ref{thm:vdW-sig}. } \label{sec:compl-pf}

Since $\Pi^\s=Q^\s P Q^\s$, by \eqref{Psig} and \eqref{Palpha'} we
obtain that
\begin{equation}\label{PPal}
\Pi^\s \doteq \sum_{\alpha \prec \prec \sigma} P^\al.
\end{equation}
Therefore, from relations \eqref{FQsigma}, \eqref{PHPstat}, \eqref{Usigmaapprox2} and \eqref{PPal} and the definition of $W^{\s,\alpha}(y)$ in \eqref{Wsig-al}{, we obtain that
\begin{equation}\label{FPEs}
F_{\Pi^\s}(E^{\s})=\sum_{\alpha \prec \prec \sigma} P^\alpha
\big(E^\s(\infty) + W^{\s,\alpha}(y)\big) 
+O(\frac{e^4}{R^7}). 
\end{equation}
Moreover, differentiating  $(P^\s H^\s P^\s-\lambda)^{-1}$ in $\lambda$ and using the second resolvent formula,
one concludes that $(P^\s H^\s P^\s-\lambda)^{-1}$ is increasing in
$\lambda \in (-\infty, \Einfty+\gamma)$, where $\gamma$ is the same as in 
\eqref{Hbot-low-bnd-stat}. It follows that the Feshbach map is decreasing which implies
 that $E^\s$ is the lowest eigenvalue of $F_{\Pi^\s}(E^\s)$.
 By \eqref{Palpha'} and by the fact that $P_a^\alpha P_b^\beta
\doteq 0, \forall \alpha \neq \beta$ and that $P_a^\alpha$ commutes
with $Q^\s$ for all $a \in \mathcal{A}^{at}, \alpha \prec \prec
\sigma$, we obtain that $P^\alpha P^\beta \doteq 0$ for all $\alpha
\neq \beta$ which together with \eqref{FPEs} and the fact that
$E^\s$ is the lowest eigenvalue of $F_{\Pi^\s}(E^{\s})$ gives that
$E^{\s}= E^\s(\infty)+ W^\s(y),$ 
where $W^\s(y)$ is defined in \eqref{Wsig} as desired

\paragraph{Proof of the necessity of Property (E).}\label{statenecessity}
In this section we will show that if Property (E) fails to hold then
so does the van der Waals - London law. To do that we modify the
analysis in Section \ref{enecessity} appropriately. Let
$E_{\min}=\min_{a \in \mathcal{A}, \alpha \prec \sigma} \inf
\sigma(H_a^{\alpha}). $ 
  We denote by
$\mathcal{A}^{\min}$ the set of all $a \in
\mathcal{A}$ for which 
$\min_{\alpha \prec \sigma }\inf \sigma(H_a^{\alpha})=\E_{\min}.$ 
Property (E) holds if and only if
$\mathcal{A}^{\min}=\mathcal{A}^{at}$. For any $a \in 
\mathcal{A}^{\min}$ we say that $\alpha \prec \prec \sigma$ if $\inf
\sigma(H_a^{\alpha})=\E_{\min}.$
 Now assume that Property (E)
fails. To prove that the van der Waals London law fails we use, as
before, the Feshbach map but with the orthogonal projection $P$
defined as the projection on
\begin{equation*}
\text{span}\{\Ran P_{a, R}^{\al}: a \in \mathcal{A}^{\min}, \alpha
\prec \prec \sigma\}, 
\end{equation*}
spanned of by the cut-off ground states of the different $H_a^\alpha$. 
 Note that the condition $\inf(H_a)=E_{\min}$  implies that $E_{\min}$ is an
isolated eigenvalue of the Hamiltonians $H_a^\alpha$ and the
eigenfunctions are exponentially decaying. The argument is the same
as in the case without statistics in Section \ref{enecessity}. We
have that
$ P=\sum_{a \in \mathcal{A}^{min}, \alpha \prec \prec \sigma}  P_{a, R}^{\al}. $ 
Proceeding similarly as in the proof of \eqref{PHPstat}, 
we obtain that
\begin{equation}
PHP \doteq E_{\min} P+ \sum_{a \in \mathcal{A}^{min}, \alpha \prec \prec \sigma}  P_{a, R}^{\al} I_a  P_{a, R}^{\al}. 
\end{equation}
The fact that $U^\sigma=O(\frac{1}{R^2})$, can be proven in the same way
as in Section \ref{enecessity}. Since Condition (E) fails we pick an
$a \in A^{\min}/ A^{at}$ and any $\alpha \prec \prec \sigma$. For
such a decomposition $a=(I_1, \dots, I_M)$, we define the charges
$q_{ai}:=(Z_i-|I_i|)e$. Taking Taylor expansion of $I_a$ (with
remainder of second order) one can show that
\begin{equation*}
P_{a, R}^{\al} I_a  P_{a, R}^{\al}= [\sum_{i\ne j}\frac{q_{ai} q_{aj}}{|y_{i}-y_{j}|}+O(\frac{1}{R^2})] P_{a, R}^{\al}.
\end{equation*}
 The rest of the proof works as in Section \ref{enecessity}.

\appendix
\bigskip

\section{More about Property (E)} \label{sec:propE}
In this appendix we prove several   statements  about Property (E) formulated in the introduction. We begin with 
\begin{proposition}[Property (E) for hydrogen atoms] \label{prop:ConditionEhyd}
 Property (E) holds  for a system of several hydrogen atoms.
  \end{proposition}
  \begin{proof}
Let 
$E^{(m)}$,  $m \geq 0$, be the ground state energy of the hydrogen ion (or atom, if $m = 0$) with charge $-me$, i.e. the lowest eigenvalue the Hamiltonian 
\begin{equation}\label{hydrogenion}
H^{(m)}=\sum_{j=1}^{m+1} H_j+\sum_{i<j}^{1, m}\frac{e^2}{|x_i-x_j|},
\end{equation}
where $H_j=-\Delta_{x_j}-\frac{e^2}{|x_j|}$ is the Hamiltonian for the hydrogen atom in the j-th coordinate.  
Property (E) 
is reduced to the property that 
for any $m\ge 1$, $E^{(m)}$ satisfies 
\begin{equation}\label{hyd-ineq}
E^{(m)}> (m+1)E^{(0)}. 
\end{equation}

Let $m_*:=\max\{m'\le m |\ H^{(m')}\ \mbox{has a ground state} \}$
and let $\psi^{(m_*)}$ be the ground state of $H^{(m_*)}$
corresponding to $E^{(m_*)}$. (We know from \cite{Hi} that $m_*\ge
1$ for $m\ge m_*$.) By the definition of $m_*$ we have $E^{(m)}
=E^{(m_*)} $ and therefore $E^{(m)}= \lan \psi^{(m_*)}, H^{(m_*)}
\psi^{(m_*)} \ran$.
Since $H_j \geq E_0$ for any $j$, we have 
$E^{(m)}> (m+1)E^{(0)}+\delta$, with
$ \del:=\lan \psi^{(m_*)}, \sum_{i<j}^{1, m_*}\frac{e^2}{|x_i-x_j|}
\psi^{(m_*)} \ran .$
\end{proof}

\DETAILS{Below we show that Property (E) follows from the following stronger
condition
\begin{itemize}
\item[(E')] We consider any two nuclei $i$ and $j$, $i \neq j$, in our system.
For any integers $m,n \ge 0$, $l > 0$ satisfying $m\le Z_i,\ m+l \leq Z_j$,
 we have the following energy inequalities
\[E_{i,m}+E_{j,-n}< E_{i,m+l}+E_{j,-n-l}.\]
\end{itemize} 
\medskip}

Next, we show that 
\begin{proposition}\label{prop:E'implE}
 Property (E') implies  Property (E).
  \end{proposition}
\DETAILS{First, we show that Property (E) implies that
\begin{equation}\label{Econsequence}
E_{1,n_1}+...+ E_{M,n_M} > E_{1}+...+ E_{M},
\end{equation}
 for all integers  $n_1, n_2,...,n_M$ satisfying
\begin{equation}\label{ni-cond} n_1+...+n_M=0,\
n_j \leq Z_j, \forall j=1,...,M,\ |n_1|+|n_2|+...+|n_M| \neq 0.
\end{equation}}
 \begin{proof}  We prove (E) 
by induction in 
the number of the atoms 
$k$. For $k=2$, Property (E) follows immediately from Property (E').
We assume that  (E) 
 holds for $M$, replaced by $k-1$ and show it holds for $M=k$.
Indeed, let $n_1,...,n_k$ be numbers satisfying the assumptions of
the Property (E) 
for $M=k$. 
By relabelling the nuclei, if necessary, we can assume  that $|n_1|
\geq |n_k|$ and $n_1
n_k<0$. 
 By Property (E'), we have that $E_{1,n_1}+E_{k,n_k}> E_{1,n_1+n_k}+E_k$.   Therefore,
$E_{1,n_1}+...+E_{k-1,n_{k-1}}+E_{k,n_k}>E_{1,n_1+n_k}+E_{2,n_2}+...+E_{k-1,n_{k-1}}+E_k$,
which together with the induction hypothesis implies  Property (E) 
 for $M=k$.
  \end{proof} 
This proves properties (a) and (b) of the introduction. For (c), it follows from the fact that $E^{(0)}$ remains the same whereas $E^{(m)}$ increases if the statistics is taken into account.

\section{Factorization of $Q_a^\alpha$} 
\label{sec:sym-gen-app}
\begin{lemma} \label{lem:Qa-prd} The projection $Q_a^\alpha$ 
is  factorised into the projections $Q_{A_j}^{\alpha_j}$ onto the multiple of irreducible representations of $S(A_j)$  of types $\alpha_j$, 
\begin{equation}\label{Qaalpha-factor}
Q_a^\alpha=\prod_{j=1}^M Q_{A_j}^{\alpha_j}.
\end{equation}
\end{lemma}
\begin{proof}
 We have that
$S(a)=\otimes_{j=1}^M S(A_j)$, where $A_j$ are the clusters of the
decomposition $a$, and $S(A_j)$ is the permutation group of the set
$A_j$. 
 We have that 
   $S(a) \ni \pi=\pi_1...\pi_M$, with $\pi_i\in S(A_i)$, and
\begin{equation*}
\# S(a)=\prod_{j=1}^M \# S(A_j)\ \mbox{ and }\
T_\pi=T_{\pi_1}...T_{\pi_M}.
\end{equation*}
The last relation and the definition of characters imply that
\begin{equation}\label{chiprod}
\chi_{\pi^{-1}}^\alpha=\prod_{j=1}^M \chi_{\pi_j^{-1}}^{\alpha_j}.
\end{equation}
 The last two relations, 
and the formula \eqref{Pbet} give \eqref{Qaalpha-factor}.
\end{proof}

\section{Lower bound on $H^{\sigma \bot}$.} \label{sec:Hbot-low-bnd-stat}
In this appendix we prove the estimate \eqref{Hbot-low-bnd-stat}. We will follow the analysis of Section \ref{Hbotbndseveral} 
modifying it appropriately.  Recall the notation used in the main text. Let  $E^\s_1(\infty)$ denote the first
excited state energy of the system of non interacting atoms and let
\begin{equation}\label{gammasigma12}
\imsgap_1^{\sigma}=\min 
\inf \sigma(H_a^\alpha)-E^\s(\infty), \quad
\imsgap_2^\sigma=E^\s_1(\infty)-E^\s(\infty),
\end{equation}
where the minimum is taken over the pairs $(a, \alpha)$, satisfying either $a \in \mathcal{A}/\mathcal{A}^{at}, \alpha \prec \sigma$ or $a \in \mathcal{A}^{at}, \alpha$ not ${\prec \prec} \sigma$. 
By Property (E) and the HVZ theorem we have that $\imsgap_1^\s>0$
and, by the definition, we also have that $\imsgap_2^\sigma>0$. We
also define
\begin{equation}\label{gammasigma0}
\imsgap_0^\sigma=\min\{\imsgap_1^{\sigma},\imsgap_2^{\sigma}\}.
\end{equation}
\begin{lemma}\label{lem:Hbot-low-bnd-stat}
There exists $\gamma^\sigma>0$ and $C>0$ such that
\begin{equation}\label{Hbot-low-bnd-stat'}
H^{ \sigma \bot} \geq (E^{\s}(\infty)+  \imsgap^{\sigma}_0-\frac{C}{R}) Q^\s.
\end{equation}
\end{lemma}
\begin{proof}
In this proof we omit for simplicity the superindex
$\sigma$ in $
P^\s, Q^\s, E^\s(\infty), \gamma_j^\s$ and write, instead, $
P, Q, E(\infty), \gamma_j$. Using that $H^{\s \bot}=P^\bot H^\s
P^\bot$ and that $Q$ commutes with $P$ we obtain that $H^{\s \bot}=Q
P^\bot H P^\bot Q$. Repeating the arguments of the proof of
\eqref{Hbotgenest1} we can obtain that
\begin{align}\label{Hbot-low-bnd-stat1}
H^{ \sigma \bot}  &\ge \sum_{a \in \mathcal{A}}Q P^\bot[  J_a H_a J_a + O(\frac{1}{R})] P^\bot Q. 
\end{align}
Using  that $Q$ commutes with $P^\bot$,  the relations $Q=Q \sum_{\alpha \prec \sigma} Q_a^\alpha$ (see  \eqref{Qsig-deco}) and $H_a  Q_a^\alpha=H_a^\alpha Q_a^\alpha$ and the fact that $Q_a^\alpha$ commutes with $J_a$ and $P^\bot$, we obtain that
\begin{align} 
Q P^\bot J_a H_a J_a P^\bot Q &=Q P^\bot J_a \sum_{\alpha \prec \sigma}   H_a^\alpha J_a P^\bot Q \notag \\ 
 \label{Hbot-low-bnd-stat2} &\ge  \sum_{a \in \mathcal{A}, \al\prec \s}Q P^\bot[  J_a H_a^\al J_a + O(\frac{1}{R})] P^\bot Q.
\end{align}
Now, we estimate $Q P^\bot J_a H_a^\al  J_a P^\bot Q$. 

\textbf{Case 1:   either $a \in \mathcal{A}/\mathcal{A}^{at}, \alpha \prec \sigma$ or $a \in \mathcal{A}^{at}, \alpha$ not ${\prec \prec} \sigma$.} 
  By \eqref{gammasigma12} and \eqref{gammasigma0},  we have $ H_a^\alpha \geq (E(\infty)+\gamma_1)Q_a^\alpha$, which, together with the previous inequality, implies
\begin{align}\label{JaHaJa-stat1}
Q P^\bot J_a H_a^\alpha J_a P^\bot Q &\geq (E(\infty)+\gamma_1) 
 Q P^\bot J_a  Q_a^\alpha  J_a P^\bot Q, 
\end{align}
\DETAILS{ The last equation together with \eqref{firstexcitedstat} and that
$P^\bot \leq 1$ gives 
\begin{equation}\label{JaHaJa-stat}
Q P^\bot J_a H_a J_a P^\bot Q \dot\geq (E(\infty)+ \gamma_0)
Q  J_a^2  Q.
\end{equation}}

\textbf{Case 2:  $a \in \mathcal{A}^{at}, \alpha {\prec \prec} \sigma$.} 
  By \eqref{Omegadef} and \eqref{partunprop1} and the support properties of $\Psi_b$, we have $ P_{b, R}^{\beta} J_a= \delta_{a, b}P_{a, R}^{\beta}$. This and  \eqref{Psig} give $  P^\bot J_a  =P_{a, R}^\bot J_a$, where, recall,  $P_{a, R} :=\sum_{\alpha \prec \prec \sigma} P_{a, R}^\alpha$. This 
  and the relation $Q_a^\al  P_{a, R}^{\beta}= \delta_{\al \beta} P_{a, R}^{\al}$ imply, after commuting $J_a$'s outside, (cf. \eqref{IMS5}) 
\DETAILS{\begin{align}
Q P^\bot J_a H_a^\al J_a P^\bot Q &=\sum_{\al \prec \sigma} 
Q_a^\al  P_{a, R}^\bot    J_a H_a^\al J_a P_{a, R}^\bot   \sum_{\al \prec \sigma} 
Q_a^\al.
\end{align}
Commuting $J_a$'s outside and using $Q_a^\al H_a =H_a^\al $, we find}
\begin{align}
Q P^\bot J_a H_a^\al J_a P^\bot Q &=
 J_a  (\1- P_{a, R}^{\al}) H_a^\al (\1-\  P_{a, R}^{\al})  J_a.
\end{align}
Using that $\|P_a-P_{a, R}\|\dot = 0$ (see \eqref{eigenf-differ}), we pass in this relation from $P_{a, R}^{\al}$ to  the orthogonal projection  $P_{a}^{\al}$  onto the ground state eigenspace of $H_a^\al$. 
\begin{align}\label{JaHaJa-stat2} Q P^\bot J_a H_a^\al J_a P^\bot Q 
&\doteq  
 J_a H_a^\al (\1 - P_{a}^{\al}) J_a. \end{align}
Then   using  
 \eqref{gammasigma0}, 
we obtain furthermore
\begin{align}\label{JaHaJa-stat3} Q P^\bot J_a H_a^\al J_a  P^\bot Q  &\geq (\Einfty+ \gamma_2) 
 J_a Q_a^\alpha  (\1 - P_{a}^{\al}) J_a .\end{align}
 Using $J_a P^\bot =P_{a, R}^\bot  J_a $ and  using \eqref{P} to go back  from $P_{a}^{\al}$ to $P_{a, R}^{\al}$,  we have 
$$ P_{a}^{\al}   J_a P^\bot \doteq  P_{a, R}^{\al}   J_a P^\bot 
 =  P_{a, R}^{\al}  P_{a, R}^\bot   J_a=0 . $$ 
Applying $P^\bot$ on both sides to \eqref{JaHaJa-stat2} and using  the last equation implies 
 \begin{align}\label{JaHaJa-stat}
Q P^\bot J_a H_a^\alpha J_a P^\bot Q &\geq (E(\infty)+\gamma_2) Q P^\bot Q_a^\alpha  J_a^2 P^\bot Q,
\end{align}
in this case.  This shows  the equation  
$Q P^\bot J_a H_a^\alpha J_a P^\bot Q \geq (E(\infty)+\gamma_0) Q P^\bot  Q_a^\alpha J_a^2 P^\bot Q,$ for all $a \in \mathcal{A}$ and $\al \prec \s$, which, together  with   \eqref{Hbot-low-bnd-stat2}, the relation
$Q=\sum_{\al \prec \sigma} Q_a^\al Q$ and the fact that $Q_a^\al$ commutes with $J_a$ and $ P$, implies  \eqref{Hbot-low-bnd-stat'}.
  \end{proof}


\section{
Supplement. Bounds 
for boosted hamiltonians} \label{sec:Resdelta-est}
In this supplement we prove bounds on the resolvent of  boosted hamiltonians, not used in this paper, but which could be useful. (In particular, similar bounds are used in \cite{Anap}.) Let  
\begin{equation}\label{Hdelta}
\H_{\delta}^\bot:= e^{-\delta \varphi ( x)} \H^\bot e^{\delta
\varphi ( x)},
\end{equation}
where $x=(x_1,...,x_N)$ is the collection of the electron
coordinates and  $\varphi ( x)$ is a $C^2$ function, with uniformly bounded derivatives up to the second order, which  is constant on the support of $\Psi_b$ for $b \neq
a$. 
\begin{proposition}\label{prop:Resdelta-est} For $R$ large enough and $\delta$ small enough (depending on $\|\nabla \varphi \|_{L^\infty} + \|\Delta \varphi \|_{L^\infty}$),
  $E$ is in the resolvent set of $\H_{\delta}^\bot$ and
 \begin{equation}\label{Resdelta-est}
\|(\H_{\delta}^\bot-E)^{-1}\| \lesssim 1.
\end{equation}
\end{proposition} 
\begin{proof} The proof consists of two lemmas.  Recall that $\E=E(y)$ is the ground state energy of $H$ and $\Delta=\sum_{j=1}^N \Delta_{x_j}$. 
 For any operator $K$, $\delta >0$ and a decomposition $a$, we let
\begin{equation}\label{Adelta}
K_{\delta}:= e^{-\delta \varphi x)} K e^{\delta \varphi (x}, \text{ } K_\delta^\bot:=(K^\bot)_{\delta}.
\end{equation}
\begin{lemma}\label{PdeltaminusPlemma}
The following inequalities hold for small $\delta$:
\begin{equation}\label{1minusdeltaPdeltaminusP}
\| P_\delta-P\| \lesssim \delta, \quad \|H_{\delta} (P_\delta-P)\| \lesssim \delta, \quad \|H (P_{\delta}-P)\|
\lesssim \delta.
\end{equation}
\end{lemma}
\begin{proof}  
Since, by the assumptions,  $\varphi(x)$ is constant on the support of $\Psi_b$ for $b \neq
a$.  This implies that
$ (P_{\Psi_b})_\delta = P_{\Psi_b}, \quad \forall b \neq a,$ 
which together with \eqref{P}, gives that
\begin{equation}\label{Pdel-P}
g(\delta):=P_\delta-P=(P_{\Psi_a})_{\delta}-P_{\Psi_a}.
\end{equation}
Clearly,  $g(\delta)$ is differentiable and 
\begin{equation}
g'(\delta)=-\vphi(x) e^{-\delta \vphi(x)} P_{\Psi_a} e^{\delta \vphi(x)}+ e^{-\delta \vphi(x)}
P_{\Psi_a} \vphi(x) e^{\delta \vphi(x)}.
\end{equation}

    Due to the exponential decay of $\Psi_a$, it follows that
$g'(\delta)$ is uniformly bounded for small $\delta$ and, since
$g(0)=0$, by applying the fundamental theorem of calculus, we obtain $\|g(\delta)\| \lesssim \delta$, which 
implies the first inequality in \eqref{1minusdeltaPdeltaminusP}.

To prove $\|H_{\delta} (P_\delta-P)\| \lesssim \delta$,  
 let $d>0$ be a constant such that $d+H_a \geq 1>0$. Using the relations
$\|H_{\delta}(d+H_b)^{-1}\| \lesssim 1$,  
$(H_a-\Einfty) P_{\Psi_a}\dot =0$ and commuting $(H_a-\Einfty)$ through $e^{-\delta \varphi(x)}$ in $(P_{\Psi_a})_{\delta}:=e^{-\delta \varphi(x)} P_{\Psi_a}e^{\delta \varphi(x)}$,  we obtain that
\begin{align}\label{leqdelta1}
\|H_{\delta}(P_\delta-P)\| &\lesssim \| (d+H_a) \left((P_{\Psi_a})_{\delta}-P_{\Psi_a} \right) \| 
\doteq \|(d+\Einfty) \left ( (P_{\Psi_a})_{\delta}-P_{\Psi_a} \right) \|\notag\\ &
+ \|  \nabla (e^{-\delta \varphi(x)}) \cdot \nabla
P_a  e^{\delta \varphi(x)}\|+ \| \Delta (e^{-\delta \varphi(x)})  P_{\Psi_a} e^{\delta \varphi(x)} \|.
\end{align}
Therefore, using \eqref{Psiadecay} and the first inequality in
\eqref{1minusdeltaPdeltaminusP} we obtain that
$\|H_{\delta} (P_\delta-P)\| \lesssim \delta$, as desired. The inequality
$\|H (P_{\delta}-P)\| \lesssim \delta$ can be proven similarly. \end{proof}
By \eqref{Hbotbnd} and  \eqref{ineqE}, the operator $H^{\bot}-E$,   where, recall,  $H^{\bot}=P^{\bot} H  P^{\bot}$, has a bounded inverse, provided  $R$ large enough, 
\begin{equation}\label{Hbotest}
\|(H^\bot-E)^{-1}\| \lesssim 1.
\end{equation}
\begin{lemma}\label{HdeltaPdeltaP}
We have that
\begin{equation}\label{Hdelbot-est3}
\|(H_\delta^\bot-H^\bot)(H^\bot-E)^{-1}\| \lesssim \del.
\end{equation}
\end{lemma}
\begin{proof} 
Observe that 
$H_{\delta}-H=-\Delta_{\delta}+\Delta.$ 
    Since $\Delta_{\delta}=e^{\delta \varphi}\Delta e^{-\delta
\varphi}$, by the
Leibnitz rule we obtain 
$-\Delta_{\delta}+\Delta=\delta [(\Delta \varphi)+ (\nabla \varphi)
\cdot \nabla-\delta |\nabla \varphi|^2].$ 
Since $\varphi$ by definition has $L^{\infty}$ bounded derivatives
this implies that $[(\Delta \varphi)+ (\nabla \varphi) \cdot
\nabla-\delta |\nabla \varphi|^2]$ is $-\Delta$ bounded. Hence, we
obtain 
\begin{equation}\label{HdeltaminusH-bnd1}
\|(H_{\delta}-H)(1-\Delta)^{-1}\| \lesssim \delta.
\end{equation}

Now, since $H^\bot:=P^\bot H P^\bot$ and $H_\delta^\bot:=(H^\bot)_{\delta}$
and $(P_\delta^\bot-P^\bot)=P-P_{\delta}$, we obtain that
\begin{equation}\label{Hdeltadifdecompo}
(H_{\delta}^\bot-H^\bot)(H^\bot-E)^{-1} =K_1+K_2+K_3,
\end{equation}
where $K_1:=-P_\delta^\bot H_\delta  (P_\delta-P) (H^\bot-E)^{-1},\ K_2:=P_{\delta}^\bot (H_{\delta}-H) P^\bot  (H^\bot-E)^{-1},\ K_3:=-(P_\delta-P) H P^\bot (H^\bot-E)^{-1}$
The terms $K_1$ and $ K_3$   are estimated by the second and third inequality in \eqref{1minusdeltaPdeltaminusP},
  respectively, and the term $K_2$ 
   is estimated by \eqref{HdeltaminusH-bnd1}   and the bound 
\begin{equation}\label{DelH-bnf} 
\|(1-\Delta)(H^\bot +C)^{-1}\| \lesssim 1,
\end{equation}
where  $C$ is such that $H^\bot +C\ge 1$, which follows
from the fact that the Coulomb potential is bounded relative to
Laplacian with the relative bound zero. As a consequence we obtain \eqref{Hdelbot-est3}.
\end{proof}

   Now, to prove  \eqref{Resdelta-est}, 
   we use the decomposition $\H_{\delta}^\bot-E=\H^\bot-E+(\H_{\delta}^\bot-\H^\bot)$ and
\eqref{Hbotest} to obtain that
\begin{equation}\label{HdeltaminusE}
(H_\delta^\bot-E) (H^\bot-E)^{-1}=(I+(H_\delta^\bot-H^\bot)(H^\bot-E)^{-1}).
\end{equation}
By the estimate \eqref{Hdelbot-est3}, we can take $\delta$ small enough to obtain that 
$\|(H_\delta^\bot-H^\bot)(H^\bot-E)^{-1}\| \leq \frac{1}{2}.$ This shows that
$I+(H_\delta^\bot-H^\bot)(H^\bot-E)^{-1}$ is invertible and its inverse is bounded by $2$,
 which together with \eqref{Hbotest}  gives, for $\delta$ small enough, the estimate \eqref{Resdelta-est}. \end{proof}



\end{document}